\documentclass[letterpaper, 11pt]{article}

\usepackage[margin=1.0in]{geometry}
\usepackage{setspace} 
	\onehalfspace
\usepackage{microtype}

\usepackage{amsmath} \allowdisplaybreaks
\usepackage{amssymb}
\usepackage{amsthm} %
\usepackage{mathtools}

\theoremstyle{definition}
\newtheorem{definition}{Definition}

\theoremstyle{plain} 
\newtheorem{theorem}{Theorem}
\newtheorem{lemma}{Lemma}

\usepackage{multirow}
\usepackage{array}
\usepackage{caption}
\usepackage{lscape}
\usepackage{subcaption}

\usepackage{natbib}
\usepackage{bibentry}

\usepackage{booktabs}
\usepackage{graphicx}
\usepackage{appendix}
\usepackage{enumerate}

\usepackage{url}
\usepackage[bookmarks=false,hidelinks]{hyperref}

\usepackage{authblk}

\newcommand{\email}[1]{{\href{mailto:#1}{\nolinkurl{#1}}}}

\usepackage{bm}
\renewcommand{\vec}[1]{\bm{#1}}

\newcommand{\set}[1]{\bm{#1}}

\usepackage{xcolor}

\usepackage{etoolbox}
\makeatletter
\def\do#1{\@namedef{#1c}{\ensuremath{\mathcal{#1}}}}
\docsvlist{A,B,C,D,E,F,G,H,I,J,K,L,M,N,O,P,Q,R,S,T,U,V,W,X,Y,Z}
\makeatother

\newcommand{\du}{\mathop{}\!\mathrm{d}u}

\renewcommand{\bar}[1]{\mkern 1.5mu\overline{\mkern-1.5mu#1\mkern-1.5mu}\mkern 1.5mu}
\renewcommand{\hat}{\widehat}
\renewcommand{\tilde}{\widetilde}

\renewcommand{\epsilon}{\varepsilon}

\newcommand{\ASatUE}{\mathsf{A}\mathsf{Sat}\mathsf{UE}}
\newcommand{\MSatUE}{\mathsf{M}\mathsf{Sat}\mathsf{UE}}

\newcommand{\UEPEX}{\mathsf{UE}\text{-}\mathsf{PE}\text{-}\set{X}}
\newcommand{\UEPEV}{\mathsf{UE}\text{-}\mathsf{PE}\text{-}\set{V}}
\newcommand{\UEPEF}{\mathsf{UE}\text{-}\mathsf{PE}\text{-}\set{F}}

\newcommand{\PRUE}{\mathsf{PRUE}}

\newcommand{\PoA}{\mathsf{PoA}}
\newcommand{\PoSat}{\mathsf{PoSat}}

\newcommand{\normV}[1]{\left\lVert#1\right\rVert_{\set{V}}}

\newcommand{\anymapsto}{\overset{\mathrm{any}}{\longmapsto}}

\usepackage{cleveref}%
\usepackage{lipsum}%
\crefname{lemma}{Lemma}{Lemmas}
\crefname{theorem}{Theorem}{Theorems}
\crefname{corollary}{Corollary}{Corollaries}

\usepackage{tikz}
\usepackage{pgf}

\usepackage{pgfplots}
\usetikzlibrary{positioning}
\usetikzlibrary{arrows,shapes}
\usetikzlibrary{arrows.meta}

\usetikzlibrary{fit}

\title{On the Price of Satisficing in Network User Equilibria}
\author{Mahdi Takalloo}
\affil{Department of Industrial and Management Systems Engineering, University of South Florida, Tampa, FL 33620}
\author{Changhyun Kwon\thanks{Corresponding author: \texttt{chkwon@usf.edu}}}
\affil{Department of Industrial and Management Systems Engineering, University of South Florida, Tampa, FL 33620}
\date{September 6, 2019}

\begin{document}
\maketitle

\begin{abstract}
When network users are satisficing decision-makers, the resulting traffic pattern attains a satisficing user equilibrium, which may deviate from the (perfectly rational) user equilibrium. In a satisficing user equilibrium traffic pattern, the total system travel time can be worse than in the case of the PRUE. We show how bad the worst-case satisficing user equilibrium traffic pattern can be, compared to the perfectly rational user equilibrium. We call the ratio between the total system travel times of the two traffic patterns the price of satisficing, for which we provide an analytical bound. We compare the analytical bound with numerical bounds for several transportation networks.
\\[1em]
\noindent\textbf{Keywords:} bounded rationality; satisficing; user equilibrium
\end{abstract}

\section{Introduction}

Instead of assuming a perfectly rational person with a clear system of preferences and perfect knowledge of the surrounding decision-making environment, we can consider \emph{boundedly} rational persons with (1) an ambiguous system of preferences and (2) lack of complete information, following \citet{simon1955behavioral}.
When decision makers are indifferent among alternatives within a certain threshold, they are called \emph{satisficing} decision makers, opposed to \emph{optimizing} decision makers. 
The notion of satisficing was first introduced by \citet{simon1955behavioral,simon1956rational}.
Satisficing decision makers choose any alternative whose utility level is above a threshold, called an \emph{aspiration level}, even when the alternative is not optimal.
The satisficing behavior is related to the first source of boundedness---an ambiguous system of preferences.

In transportation research, modeling drivers' route choice is an important task. 
While the travel-time minimization has been traditionally used as a basis for such modeling, sub-optimal route-choice behavior has gained attention.
Since \citet{mahmassani1987boundedly}, bounded rationality has gained attention in the transportation research literature \citep{szeto2006dynamic, wu2013bounded, han2015formulation, szeto2006dynamic, ge2012alternative, di2014braess, guo2013toll, lou2010robust}. 
Empirical evidence supports bounded rationality of drivers \citep{nakayama2001drivers,zhu2010people}. 
The notion of bounded rationality has also been considered in the evaluation of value of times in connection to route-choice modeling \citep{xu2017route}, and in the model of behavior adjustment process \citep{ye2017rational}.
We refer readers to a review of \citet{di2016boundedly}.
In the non-transportation literature, the notion of bounded rationality and satisficing has also received much attention \citep{charnes1963deterministic,lam2013multiple,jaillet2016satisficing,chen1997boundedly,brown2009satisficing}.

While the above-mentioned transportation research literature considers boundedly rational drivers, their discussion is limited to satisficing drivers without considering the second source of boundedness: lack of complete information on the decision environment.
\citet{sungeneralized} connect the first and the second sources of boundedness by considering both satisficing behavior and incomplete information, in the context of shortest-path finding in \emph{congestion-free} networks.
\citet{sungeneralized} study the second source by considering errors in drivers' perception of arc travel time, and conclude that their perception-error model can generally capture both sources of boundedness in rationality in a single unified modeling framework.

In the literature, the traditional network user equilibrium, Wardrop equilibrium in particular, is called the perfectly rational user equilibrium (PRUE), while a traffic pattern equilibrated among satisficing drivers is called a boundedly rational user equilibrium (BRUE).
In this paper, we will use a new term \emph{satisficing user equilibrium} (SatUE) instead of BRUE to emphasize that it only considers the first source of boundedness without considering drivers' incomplete information on the decision environment.
We believe that the term `BRUE' should be used to describe a broader and more general class of models, including SatUE.

Note that SatUE differs from the stochastic user equilibrium (SUE) \citep{sheffi1985urban} in two important aspects.
First, drivers are assumed to be optimizing decision makers in SUE, while they are satisficing in SatUE.
Second, with appropriate probability distributions assumed in the random utility model in SUE, each path possesses a probability of being chosen; hence we can compute the expected traffic flow rate in each path. 
In SatUE, however, each satisficing path is acceptable to drivers, but it may or may not be chosen by drivers and we do not know its probability of being chosen.
See further discussion in \citet{di2016boundedly}. 

The main contribution of this paper is the quantification of how bad the total system travel time in SatUE can be.
In a SatUE traffic pattern, the total system travel time can be either greater than or less than that of PRUE. 
We define the \emph{price of satisficing} (PoSat) as the ratio between the worst-case total system travel time of SatUE and the total system travel time of PRUE.
This paper quantifies PoSat analytically and compares with numerical bounds.

The analytical quantification of PoSat is related to the price of anarchy (PoA) \citep{koutsoupias1999worst,roughgarden2002bad} that compares the performances of the system optimal solutions and the PRUE solutions. 
Using a similar idea, we can also compare the performance of the perfectly rational user equilibrium traffic patterns and satisficing user equilibrium traffic patterns. 
While PoA quantifies how much system-wide performance we can lose by competing, PoSat quantifies how much we can lose by satisficing. 
\citet{roughgarden2002bad} define and study the PoA of \emph{approximate Nash equilibria}, which are essentially SatUE patterns. 
We develop our bounds for PoSat based on the bounds for PoA of approximate Nash equilibria \citep{christodoulou2011performance} and the ideas from the sensitivity analysis of traffic equilibria \citep{dafermos1984sensitivity}.
Note that \citet{perakis2007price} studies the PoA of the \emph{exact} Nash equilibria with general nonlinear, asymmetric cost functions.

The notion of PoSat is also related to the \emph{price of risk aversion} \citep{nikolova2015burden} and the \emph{deviation ratio} \citep{kleer2016impact}.
When network users are risk-averse decision makers, the price of risk aversion compares the performances of the resulting equilibrium among risk-averse users and the (risk-neutral) PRUE.
When network users' cost functions are deviated from the true cost functions for some reasons, the deviation ratio compares the performances of the resulting equilibrium and the PRUE. 
\citet{kleer2016impact} show that the price of risk aversion is a special case of the deviation ratio. 
In both research articles, however, only cases with a common single origin node are considered.
In this paper, we consider general cases with multiple origin nodes and multiple destination nodes, with asymmetric travel time functions.

This paper is organized as follows. 
In Section \ref{sec:notation}, we introduce the notation and define various concepts including user equilibrium, system optimum, satisficing behavior, price of anarchy, and price of satisficing.
In Section \ref{sec:UE-PE}, we define the user equilibrium with perception errors and make connections with satisficing user equilibrium.
Our main result is introduced in Section \ref{sec:theoretical_PoSat}, where we derive the analytical worst-case bound on the price of satisficing. 
In Section \ref{sec:numerical_PoSat}, we compare the analytical bound with numerical bounds.
Section \ref{sec:conclusion} concludes this paper.

\section{Notation and Definitions} \label{sec:notation}

Since we will use path-based and arc-based flow variables and their corresponding functions and sets interchangeably, we need clear definitions of variables, sets, and functions. 
We use boldfaced lower-case letters for vector quantities as in $\vec{v}$ and normal lower-case letters for their components as in $v_a$; similarly, vector-valued functions like $\vec{t}(\cdot)$ and their components like $t_a(\cdot)$.
We use boldfaced upper-case letters for the set that they belong to, as in $\vec{v} \in \set{V}$.
We use calligraphic capital letters for sets of indices as in $\Nc$.
The only exception is that $\vec{Q}$ (a bold-face capital letter instead of lower case) represents a
vector of $Q_w$, the demand for OD pair $w$. The lower-case version $q^w_i$ is instead the net
amount of flow associated with OD pair $w$ that enters or leaves node $i$ in the next
subsection

\subsection{Traffic Flow Variables and Feasible Sets}
We consider a network with a set of origin and destination $\Wc$ that is represented by directed graph $G (\Nc,\Ac)$, where $\Nc$ is the set of nodes, and $\Ac$ is the set of arcs. 
For each OD pair $w\in\Wc$, the travel demand is $Q_w$ and the set of available paths is $\Pc_w$. 
The set of all available paths in the whole network is defined as $\Pc = \cup_{w\in\Wc} \Pc_w$.

We also define the set of path flow variables $\vec{f}$ as 
\begin{align*}
	\set{F} &= \bigg\{ \vec{f} :  \sum_{p\in\Pc_w} f_p = Q_w \quad \forall w\in\Wc, \qquad f_p \geq 0 \quad \forall p\in\Pc \bigg \}
\end{align*}
and the corresponding set of arc flow variables $\vec{v}$ is defined as
\[
	\set{V} = \bigg\{ \vec{v} : 	v_a = \sum_{p\in\Pc} \delta^p_a f_p \quad \forall a\in\Ac, \qquad \vec{f}\in \set{F}  \bigg\}
\]
where $\delta^p_a=1$ if path $p$ contains arc $a$ and $\delta^p_a=0$ otherwise.
Let $\Ac^+_i$ and $\Ac^-_i$ be the set of arcs whose tail node and head node are $i$, respectively.
When we need to preserve OD information in arc flow variables, we use $\vec{x}$ as follows:
\begin{align*}
	\set{X} &= \bigg\{ \vec{x} : x^w_a = \sum_{p\in\Pc_w} \delta^p_a f_p \quad \forall a\in\Ac,w\in\Wc \qquad \vec{f}\in \set{F} \bigg\} \\
					&= \bigg\{ \vec{x} : \sum_{a\in\Ac^+_i} x^w_a - \sum_{a\in\Ac^-_i} x^w_a = q^w_i \quad \forall w\in\Wc, i\in\Nc \bigg\}
\end{align*}
where $q_i^w= -Q_w$ if $i=o(w)$, $q_i^w= Q_w$ if $i=d(w)$, and $q_i^w= 0$ otherwise.

We have $v_a = \sum_{p\in\Pc} \delta^p_a f_p$, $x^w_a = \sum_{p\in\Pc_w} \delta^p_a f_p$, and $v_a = \sum_{w\in\Wc} x^w_a$.
Therefore, the transformations from $\vec{f}$ to $\vec{v}$, from $\vec{f}$ to $\vec{x}$, and from $\vec{x}$ to $\vec{v}$ are unique, which are denoted by $\vec{f}\mapsto\vec{v}$, $\vec{f}\mapsto\vec{x}$, and $\vec{x}\mapsto\vec{v}$, respectively.
The inverse transformations are, however, not unique. 
In the rest of this paper, to emphasize the non-uniqueness of the transformation and refer to \emph{any} result of such transformation, we use $\anymapsto$; for example, with $\vec{v}\anymapsto\vec{f}$, we consider any $\vec{f}$ such that $v_a = \sum_{p\in\Pc} \delta^p_a f_p$.

We will use $\vec{v}$, $\vec{f}$, and $\vec{x}$ interchangeably to describe the same traffic pattern. 
In particular, we define
\begin{itemize}
\item $\vec{f}^*$, $\vec{v}^*$, $\vec{x}^*$ : system optimal flow vectors (Section \ref{sec:travel_time})
\item $\vec{f}^0$, $\vec{v}^0$, $\vec{x}^0$ : perfectly rational user equilibrium flow vectors (Section \ref{sec:PRUE})
\item $\vec{f}^\kappa$, $\vec{v}^\kappa$, $\vec{x}^\kappa$ : (multiplicative) satisficing user equilibrium flow vectors with a multiplicative factor (to be defined
subsequently) $\kappa$ (Section \ref{sec:sat})
\end{itemize}
Note that when $\kappa=0$, we have $\vec{f}^\kappa=\vec{f}^0$.

\subsection{Travel Time Functions and System Optimum} \label{sec:travel_time}

We denote arc travel function with arc traffic volume $\vec{v}$ by $t_a(\vec{v})$ for each arc $a\in\Ac$. 
We consider a performance function for each arc $a$ as
\[
	z_a(\vec{v}) = t_a(\vec{v}) v_a .
\]
We denote the travel time function along path $p$ with flow $\vec{f}$ by $c_{p}(\vec{f})$. 
When written as functions of $\vec{x}$, the arc travel time is denoted as $\tau_a(\vec{x})$ or $\tau_a^w(\vec{x})$, where the latter is used to emphasize the focus on OD pair $w$. Of course, $\tau_a^w(\vec{x})=\tau_a(\vec{x})=t_a(\vec{v})$, where $\vec{v}=\sum_w x_a^w$. The performance function for path $p \in \Pc$ is as follows:
\[
	z_{p}(\vec{f}) = c_{p}(\vec{f}) f_{p} .
\]
The following shows the relationship between path and arc travel times.
\begin{align*}
c_p(\vec{f}) &= \sum_{a\in\Ac} \delta^p_a t_a(\vec{v}) .
\end{align*}

We define the arc-based total system performance function $Z(\vec{v})$ and path-based total system performance function $C(\vec{f})$ interchangeably as follows:
\begin{align*} 
	Z(\vec{v}) &\equiv \sum_{a\in\Ac} z_a(\vec{v}) = \sum_{a\in\Ac} t_a(\vec{v}) v_a \\
	     &= \sum_{p\in\Pc} z_p(\vec{f}) = \sum_{p\in\Pc} c_p(\vec{f})f_p  = \sum_{w\in\Wc} \sum_{p\in\Pc_w}  c_p(\vec{f}) f_p \equiv C(\vec{f}),
\end{align*}
which is also called the total system travel time.
A flow pattern that minimizes $Z(\cdot)$ or $C(\cdot)$ is called a \emph{system optimal} flow pattern.

The vector-valued function $\vec{t}(\cdot)$ is called \emph{monotone} in $\set{V}$ if
\begin{equation} \label{monotone}
	[\vec{t}(\vec{v}^1) - \vec{t}(\vec{v}^2)]^\top (\vec{v}^1 - \vec{v}^2) \geq 0 
\end{equation}
for all $\vec{v}^1, \vec{v}^2 \in \set{V}$. 
If \eqref{monotone} holds as a strict inequality for all $\vec{v}^1 \neq \vec{v}^2$, it is said \emph{strictly monotone}.
The function $\vec{t}(\cdot)$ is called \emph{strongly monotone} in $\set{V}$ with modulus $\alpha>0$ if
\begin{equation} \label{strongly_monotone}
[\vec{t}(\vec{v}^1) - \vec{t}(\vec{v}^2)]^\top (\vec{v}^1 - \vec{v}^2) \geq \alpha \normV{\vec{v}^1 - \vec{v}^2}^2
\end{equation}
for all $\vec{v}^1, \vec{v}^2 \in \set{V}$, where $\normV{\cdot}$ is the $l^2$-norm in $\set{V}$.
The monotonicity of path-based travel time function $c_p(\cdot)$ or its vector form $\vec{c}(\cdot)$ can be similarly defined.
The path-based function $c_p(\cdot)$, however, is not strongly monotone in general \citep[e.g., see Example 3 in][]{de1998optimization}.

\subsection{Perfectly Rational User Equilibrium} \label{sec:PRUE}

When network users are perfectly rational---they seek the shortest path---we attain the perfectly rational user equilibrium (PRUE) defined as follows:

\begin{definition}  [Perfectly Rational User Equilibrium] 
A traffic pattern $\vec{f}^0$ is called a \emph{perfectly rational user equilibrium} (PRUE), if 
\begin{equation} \label{UE_cond}
(\PRUE) \qquad	f_p^0 >0 \implies c_p(\vec{f}^0) = \min_{p'\in\Pc_w} c_{p'}(\vec{f}^0) 
\end{equation}
for all $p\in\Pc_w$ and $w\in\Wc$.
\end{definition}

Using the arc travel function, the above condition can be restated as follows 
\begin{equation} \label{UE_cond_arc}
f_p^0 >0 \implies \sum_{a\in\Ac} \delta^p_a t_a(\vec{v}^0) = \min_{p'\in\Pc_w} \sum_{a\in\Ac} \delta^{p'}_a t_a(\vec{v}^0) 
\end{equation}
for all $p\in\Pc_w$ and $w\in\Wc$.

It is well known that a solution to the following variational inequality problem is a user equilibrium traffic flow \citep{smith1979existence,dafermos1980traffic}:
\begin{equation} \label{ue-vi-f}
\text{to find } \vec{\bar{f}} \in \set{F} : \sum_{p\in\Pc} c_p(\bar{\vec{f}}) (f_p - \bar{f}_p)  \geq 0 \qquad \forall \vec{f}\in \set{F},
\end{equation}
which can be equivalently rewritten as: 
\begin{equation} \label{ue-vi-v}
\text{to find } \vec{\bar{v}} \in \set{V} : \sum_{a\in\Ac} t_a(\bar{\vec{v}}) (v_a - \bar{v}_a) \geq 0 \qquad \forall \vec{v}\in \set{V},
\end{equation}
 or
\begin{equation} \label{ue-vi-x}
\text{to find } \vec{\bar{x}} \in \set{X} : \sum_{a\in\Ac} \sum_{w\in\Wc} \tau_a(\bar{\vec{x}}) (x^w_a - \bar{x}^{w}_a) \geq 0 \qquad \forall \vec{x}\in \set{X}
\end{equation}
where $\tau^w_a(\vec{x}) = \tau_a(\vec{x}) = t_a(\vec{v})$.

With strictly monotone functions $t_a(\cdot)$, the solution $\bar{\vec{v}}$ to \eqref{ue-vi-v} is unique.
While the transformations $\bar{\vec{v}}\anymapsto\bar{\vec{f}}$ and $\bar{\vec{v}}\anymapsto\bar{\vec{x}}$ are not unique, any such $\bar{\vec{f}}$ and $\bar{\vec{x}}$ are solutions to \eqref{ue-vi-f} and \eqref{ue-vi-x}, respectively; therefore, solutions to \eqref{ue-vi-f} and \eqref{ue-vi-x} are not unique in general.

When the travel time on arc $a$ is a function of only $v_a$, i.e. $t_a=t_a(v_a)$, then it is called \emph{separable}.
With separable arc travel time functions, the variational inequality problem \eqref{ue-vi-v} admits an equivalent convex optimization problem as formulated by \citep{beckmann1956studies}.
In general, if the Jacobian matrix of the arc travel time function vector $\vec{t}(\vec{v})$ is symmetric, that is,
\[
	\frac{\partial t_a(\vec{v})}{\partial v_e} = \frac{\partial t_e(\vec{v})}{\partial v_a} \qquad \forall a,e\in\Ac,
\]
for all $\vec{v}\in\set{V}$,
the variational inequality problem \eqref{ue-vi-v} can be reformulated as an equivalent Beckmann-type convex optimization problem \citep{patriksson2015traffic,friesz2016foundations}.
When the Jacobian is asymmetric, no Beckmann-type convex optimization problem
equivalent to \eqref{ue-vi-v} exists in general. In this case, the arc travel time functions is characterized
as asymmetric and obtaining a PRUE flow requires solving a variational inequality problem.

\subsection{Satisficing User Equilibrium} \label{sec:sat}

We introduce definitions of satisficing behavior and corresponding user equilibrium traffic patterns.
In transportation research literature, boundedly rational user equilibrium (BRUE) is
often defined with an additive term \citep[see e.g.,][]{lou2010robust, di2013boundedly, han2015formulation}. Herein, we refer to BRUE in the literature as an `additive satisficing user
equilibrium' to (i) highlight its additive feature and (ii) limit user behavior to just
satisficing. 
Bounded rationality includes behaviors other than satisficing as well.

\begin{definition}  [Additive Satisficing] \label{def:ASat}
A traffic pattern $\vec{f}$ is called an \emph{additive satisficing user equilibrium (ASatUE)} with an additive factor $E$ , if 
\begin{equation} \label{ASat_cond}
(\ASatUE) \qquad	f_p >0 \implies c_p(\vec{f}) \leq \min_{p'\in\Pc_w} c_{p'}(\vec{f}) + E
\end{equation}
for all $p\in\Pc_w$ and $w\in\Wc$, where $E$ is a positive constant.
\end{definition}

We can also derive a similar definition using a multiplicative term.
While the additive form in Definition \ref{def:ASat} is popularly used in the transportation research literature, the multiplicative form in Definition \ref{def:MSat} enables us to consider the satisficing level in disaggregate arc levels as we will observe in this paper.
Multiplicative satisficing user equilibrium is also called approximate Nash equilibrium in the price of anarchy literature \citep{christodoulou2011performance}.

\begin{definition}  [Multiplicative Satisficing] \label{def:MSat}
A traffic pattern $\vec{f}^\kappa$ is called a \emph{multiplicative satisficing user equilibrium} with a multiplicative factor $\kappa$, or $\kappa$-MSatUE, if 
\begin{equation} \label{MSat_cond}
(\MSatUE) \qquad	f_p^\kappa >0 \implies c_p(\vec{f}^\kappa) \leq (1+\kappa) \min_{p'\in\Pc_w} c_{p'}(\vec{f}^\kappa) 
\end{equation}
for all $p\in\Pc_w$ and $w\in\Wc$, where $\kappa \geq 0$ is a constant.
\end{definition}

Note that the additive ($E$) and multiplicative ($\kappa$) factor in \eqref{ASat_cond} and \eqref{MSat_cond}, respectively, may be
defined for each OD pair $w$. For example, $E_w$ and $\kappa_w$ can replace $E$ and $\kappa$ in \eqref{ASat_cond} and \eqref{MSat_cond},
respectively, to allow for non-homogeneous satisficing thresholds.
In such cases, however, we assume that travelers for the same OD pair are homogeneous with the same threshold $E_w$ or $\kappa_w$.
In this paper, to describe the satisficing behavior, we focus only on $\MSatUE$. Moreover, for simplicity, we use a single value of $\kappa$ for all OD pairs.

\subsection{Price of Satisficing}
The price of anarchy (PoA) compares the performances of a satisficing user equilibrium ($C(\vec{f}^\kappa)$) against that of a system optimum ($C(\vec{f}^*)$).
Among possibly multiple satisficing user equilibrium traffic patterns, we are interested in the worst-case. 
Let $\Psi_\kappa(G, \vec{Q}, \vec{t})$ be the set of all satisficing user equilibria with a multiplicative factor $\kappa$
where $G, \set{Q}$, and $\vec{t}$ denote the underling network, demand vector, and travel time function,
respectively. Then, the PoA for the triplet $(G, \set{Q}, \vec{t})$ is defined as follows:
\begin{equation} \label{PoA}
	\PoA_\kappa(G, \vec{Q}, \vec{t}) = \max_{\vec{f}^\kappa \in \Psi_\kappa(G, \vec{Q}, \vec{t})} \frac{C(\vec{f}^\kappa)}{C(\vec{f}^*)},
\end{equation}
where $\vec{f}^*$ is the system optimum flow for $(G, \vec{Q}, \vec{t})$.
We are usually interested in its upper bound over a set of triplets, $\Omega$., i.e. 
\[
	\sup_{(G, \vec{Q}, \vec{t}) \in \Omega} \PoA_\kappa(G, \vec{Q}, \vec{t})
\]

In the context of bounded rationality and satisficing, we are more interested in comparing the performance of approximate Nash equilibrium $C(\vec{f}^\kappa)$ and the performance of the perfectly rational user equilibrium $C(\vec{f}^0)$.
We define the price of satisifcing ($\PoSat$) of instance $(G, \vec{Q}, \vec{t})$ as follows:
\begin{equation} \label{PoSat}
	\PoSat_\kappa(G, \vec{Q}, \vec{t}) = \max_{\vec{f}^\kappa \in \Psi_\kappa(G, \vec{Q}, \vec{t})} \frac{C(\vec{f}^\kappa)}{C(\vec{f}^0)},
\end{equation}
and its upper bound over $\Omega$ is
\[
\sup_{(G, \vec{Q}, \vec{t}) \in \Omega} \PoSat_\kappa(G, \vec{Q}, \vec{t})
\]

In this paper, $\Omega$ is a set of all triplets where $G$ is a directed graph with a finite number
of nodes and arcs, $\set{Q}$ is a vector of finite and positive constants, and $\vec{t}(\cdot)$ is a vector of
polynomial functions with nonnegative coefficients and of order $n \ge 0$. To emphasize
the latter, we also write $\Omega(n)$ instead of $\Omega$ when appropriate.

\section{User Equilibrium with Perception Errors} \label{sec:UE-PE}

Related to $\MSatUE$ is the user equilibrium with perception error (UE-PE) model.
In this model, we assume that network users are optimizing, i.e. seeking the shortest path; however, we assume that users may have their own perception of the travel time function.
 
We let $\epsilon^w_a$ denote the perception error of travel time along arc $a$ of users in OD pair $w$.
A vector $\vec{\bar{x}}\in \set{X}$ is a solution to the UE-PE model, if 
\begin{equation} \label{ue-pe}
\sum_{a\in\Ac} \sum_{w\in\Wc} ( t_a(\bar{\vec{v}}) - \epsilon^w_a ) (x^w_a - \bar{x}^{w}_a) \geq 0 \qquad \forall \vec{x}\in \set{X}
\end{equation}

The above variational inequality assumes that $\vec{\epsilon}$ is sufficiently small, i.e., $0 \le \epsilon_a^w < t_a(\bar{v})$ for all $a \in \Ac, w \in \Wc$.
Under such an assumption, $t_a(\vec{\bar{v}})-\epsilon_a^w$ can be viewed as the \emph{perceived} travel time for arc $a$ for drivers of OD pair $w$.

The term $\epsilon^w_a$ represents the perception error for arc $a$ and OD pair $w$. In this model, we assume all drivers for each OD pair are homogeneous in their perception of arc travel time.

With changes of variables $\lambda^w_a t_a(\vec{v}) = t_a(\vec{v}) - \epsilon^w_a$, the UE-PE model \eqref{ue-pe} can be restated as follows:
\begin{equation} \label{ue-pe-x}
(\UEPEX) \qquad 	\sum_{a\in\Ac} \sum_{w\in\Wc} \lambda^w_a t_a(\bar{\vec{v}}) (x^w_a - \bar{x}^{w}_a) \geq 0 \qquad \forall \vec{x}\in \set{X}
\end{equation}
for some $\vec{\lambda}$ such that $\lambda_a^w \in (0,1]$ for all $w\in\Wc$ and $a\in\Ac$. 
We observe that the UE-PE model generates a subset of MSatUE traffic flow patterns.

\begin{lemma}[$\UEPEX \implies \text{MSatUE}$] \label{lem:ue-pe-x}
Suppose $\bar{\vec{x}}$ is a solution to $\UEPEX$ in \eqref{ue-pe-x} with some $\vec{\bar{\lambda}}$ where $\bar{\lambda}^w_a\in[\frac{1}{1+\kappa},1]$ for all $w\in\Wc$ and $a\in\Ac$. 
Then any $\bar{\vec{f}}$ with $\bar{\vec{x}}\anymapsto\bar{\vec{f}}$ is a $\kappa$-MSatUE flow.
\end{lemma}
\begin{proof}[Proof of Lemma \ref{lem:ue-pe-x}]
Given $\bar{\vec{f}}$, we let $\bar{\vec{v}}$ be the arc flow vector from $\bar{\vec{f}}\mapsto\bar{\vec{v}}$.
Let $\bar{\epsilon}$ be the perception error that makes $\bar{\vec{x}}$ a solution to \eqref{ue-pe}. Then, $\vec{\bar{x}}$ is a user equilibrium flow with respect to arc travel time $\lambda_a^w t_a(\cdot)$ and the following
follows from \eqref{UE_cond_arc}:

\begin{equation}\label{ue_epsilon}
	\bar{f}_p>0 \implies \sum_{a\in\Ac} \delta^p_a \bar{\lambda}^w_a t_a(\bar{\vec{v}}) = \min_{p'\in\Pc_w} \sum_{a\in\Ac} \delta^{p'}_a \bar{\lambda}^w_a t_a(\bar{\vec{v}}) 
\end{equation}
for all $p\in\Pc_w$ and $w\in\Wc$.
Since $\bar{\lambda}^w_a\in[\frac{1}{1+\kappa},1]$, the right-hand-side of \eqref{ue_epsilon} implies
\[
	\frac{1}{1+\kappa} \sum_{a\in\Ac} \delta^{p}_a  t_a(\bar{\vec{v}}) 
	\leq \min_{p'\in\Pc_w} \sum_{a\in\Ac} \delta^{p'}_a \bar{\lambda}^w_a t_a(\bar{\vec{v}})  
	\leq  \min_{p'\in\Pc_w} \sum_{a\in\Ac} \delta^{p'}_a t_a(\bar{\vec{v}})  ,
\]
which is equivalent to the following path flow form:
\[
	c_{p}(\vec{\bar{f}}) 
	\leq  (1+\kappa) \min_{p'\in\Pc_w} c_{p'}(\vec{\bar{f}}) .
\]
Therefore, we conclude that $\bar{\vec{f}}$ is a $\kappa$-MSatUE traffic flow.
\end{proof}

We can also provide a path-based formulation of UE-PE:
\begin{equation}\label{ue-pe-f}
(\UEPEF)\quad	\sum_{w\in\Wc} \sum_{p\in\Pc_w} \tilde{c}_p^w(\bar{\vec{f}}) ( f_p - \bar{f}_p ) \geq 0 \qquad \forall \vec{f}\in\set{F}
\end{equation}
for the perceived path travel time functions $\tilde{c}_p^w(\vec{f})=\sum_{a\in\Ac} \delta^p_a \lambda^w_a t_a(\vec{v})$ with \emph{some} $\vec{\lambda}$ such that $\lambda_a^w \in (0,1]$ for all $a\in\Ac, w\in\Wc$.

\begin{lemma}[$\UEPEF\iff\UEPEX$] \label{lem:ue-pe-f}
If $\bar{\vec{f}}\in\set{F}$ is a solution to $\UEPEF$ in \eqref{ue-pe-f} for some $\vec{\lambda}$ such that $\lambda^w_a \in [\frac{1}{1+\kappa},1]$,
then $\bar{\vec{x}}$ with $\bar{\vec{f}}\mapsto\bar{\vec{x}}$ is a solution to $\UEPEX$ in \eqref{ue-pe-x}.
Conversely, if $\bar{\vec{x}}\in\set{X}$ is a solution to $\UEPEX$ in \eqref{ue-pe-x},
then any $\bar{\vec{f}}$ with $\bar{\vec{x}}\anymapsto\bar{\vec{f}}$ is a solution to $\UEPEF$ in \eqref{ue-pe-f}.
\end{lemma}
\begin{proof}[Proof of Lemma \ref{lem:ue-pe-f}]
We can prove both directions by observing that
\begin{align*}
\sum_{w\in\Wc} \sum_{p \in \Pc_w} \tilde{c}_p^w(\bar{\vec{f}}) ( f_p - \bar{f}_p ) 
&= \sum_{w\in\Wc} \sum_{p\in\Pc_w} \sum_{a\in\Ac} \delta^p_a \lambda^w_a t_a(\bar{\vec{v}}) ( f_p - \bar{f}_p ) \\
&= \sum_{w\in\Wc} \sum_{a\in\Ac} \lambda^w_a t_a(\bar{\vec{v}}) \bigg( \sum_{p\in\Pc_w} \delta^p_a f_p - \sum_{p\in\Pc_w} \delta^p_a \bar{f}_p \bigg) \\
&= \sum_{w\in\Wc} \sum_{a\in\Ac} \lambda^w_a t_a(\bar{\vec{v}}) ( x^w_a - \bar{x}^w_a).
\end{align*}
\end{proof}

When the values of $\lambda^w_a$ are the same across all $w\in\Wc$, i.e. $\lambda_a=\lambda^w_a$ for all $w\in\Wc$, we can simplify \eqref{ue-pe-x} as follows:
\begin{equation} \label{ue-pe-v}
(\UEPEV) \quad \sum_{a\in\Ac} \lambda_a t_a(\bar{\vec{v}}) (v_a - \bar{v}_a) \geq 0 \qquad \forall \vec{v}\in \set{V}
\end{equation}
for \emph{some} $\vec{\lambda}$ such that $\lambda_a \in (0,1]$ for each $a\in\Ac$. 
The simplified model \eqref{ue-pe-v} has been considered in the literature for approximate Nash equilibrium \citep{christodoulou2011performance} and Nash equilibrium with deviated travel time functions \citep{kleer2016impact}.
For the simplified model, we can state:
\begin{lemma}[$\UEPEV \implies \UEPEX$] \label{lem:ue-pe-v}
Suppose that $\bar{v}\in \set{V}$ is a solution to $\UEPEV$ in \eqref{ue-pe-v} for some $\vec{\lambda}$ such that $\lambda_a \in [\frac{1}{1+\kappa},1]$ for all $a\in\Ac$.
Let $\bar{\vec{x}}$ be any vector with $\bar{\vec{v}}\anymapsto\bar{\vec{x}}$.
Then $\bar{\vec{x}}$ is a solution to $\UEPEX$ in \eqref{ue-pe-x}.
\end{lemma}

While Lemmas \ref{lem:ue-pe-x}, \ref{lem:ue-pe-f}, and \ref{lem:ue-pe-v} provide sufficient conditions for a traffic flow pattern to be a $\kappa$-MSatUE, Theorem 1 of \citet{christodoulou2011performance} provides a necessary condition.
Although \citet{christodoulou2011performance} assumed separable arc travel time functions, their proof is still valid for nonseparable travel time functions.

\begin{lemma}[A necessary condition of MSatUE] \label{lem:necessary}
Let $\vec{f}^\kappa\in \set{F}$ be a $\kappa$-MSatUE and $\vec{v}^\kappa\in\set{V}$ be the corresponding arc flow vector with $\vec{f}^\kappa \mapsto \vec{v}^\kappa$.
Then we have
\begin{equation} \label{eq:necessary}
	\sum_{a\in\Ac} t_a(\vec{v}^\kappa) ( (1+\kappa)v_a - v^\kappa_a ) \geq 0 \qquad \forall \vec{v}\in\vec{V}.
\end{equation}
\end{lemma}

\citet{christodoulou2011performance} derive a tight bound on the price of anarchy on approximate Nash equilibria based on Lemma \ref{lem:necessary}. 

\begin{figure}
\begin{center}
\begin{tikzpicture}
\node at (0,0) (X) {$\UEPEX$}; 
\node[below=0.7cm of X] (F)  {$\UEPEF$};
\node[left=0.7cm of X] (V)  {$\UEPEV$};
\node[right=0.7cm of X] (M) {$\mathrm{MSatUE}$}; 
\node[right=0.7cm of M] (N) {\eqref{eq:necessary}}; 
\draw[implies-implies, double equal sign distance] (F)--(X);
\draw[-implies, double equal sign distance] (V)--(X);
\draw[-implies, double equal sign distance] (X)--(M);
\draw[-implies, double equal sign distance] (M)--(N);
\end{tikzpicture}
\caption{Summary of Lemmas \ref{lem:ue-pe-x}--\ref{lem:necessary}. The relation $\mathsf{X} \implies \mathsf{Y}$ means that any solution to $\mathsf{X}$ yields a solution to $\mathsf{Y}$.}
\label{fig:summary}
\end{center}
\end{figure}

Figure \ref{fig:summary} summarizes the relationships between problems addressed in Lemma \ref{lem:ue-pe-x} to \ref{lem:necessary}.
Note that this result does not mean that the notions of UE-PE and SatUE are equivalent in general.
Rather, the specific modeling of $\UEPEX$, $\UEPEF$, and $\UEPEV$ with the perception error interval $[\frac{1}{1+\kappa},1]$, as defined in \eqref{ue-pe-x}, \eqref{ue-pe-f}, and \eqref{ue-pe-v}, respectively, leads to the result in Figure \ref{fig:summary}.
If we use a different definition of perception error sets, such a result may not hold.

\section{Bounding the Price of Satisficing} \label{sec:theoretical_PoSat}

We first provide analytical bounds of $C(\vec{f}^\kappa)$ compared to $C(\vec{f}^0)$.

\subsection{Lessons from the Price of Anarchy}

We first observe that $\PoSat_\kappa(G, \vec{Q}, \vec{t}) \leq	\PoA_\kappa(G, \vec{Q}, \vec{t})$ for any network instance $(G, \vec{Q}, \vec{t})$, since $C(\vec{f}^0) \geq C(\vec{f}^*)$. 
This enables us to use the results from the price of anarchy literature for bounding PoSat. 
Theorem 2 of \citet{christodoulou2011performance} bound the price of anarchy when arc travel-time functions are separable and
polynomial with nonnegative coefficients and of degree $n$ that leads immediately to the
following result:

\begin{lemma} \label{lem:bound_from_poa}
Suppose $\vec{f}^\kappa$ is a $\kappa$-MSatUE flow, and $t_a(\cdot)$ is polynomial with nonnegative
coefficients and of degree $n$.
Define
\begin{equation}
\zeta(\kappa,n) =
\begin{cases}
(1+\kappa)^{n+1} & \text{ if } \kappa \geq (n+1)^{1/n} - 1 ,\\
\bigg(\frac{1}{1+\kappa} - \frac{n}{(n+1)^{(n+1)/n}} \bigg)^{-1} & \text{ if } 0 \leq \kappa \leq (n+1)^{1/n} - 1.
\end{cases}		
\end{equation}
Then we have
\begin{equation}\label{zetabound}
	C(\vec{f}^*)
		\leq 
	C(\vec{f}^\kappa)
		\leq
	\zeta(\kappa,n) C(\vec{f}^*)
		\leq 
	\zeta(\kappa,n) C(\vec{f}^0).
\end{equation}
That is, the PoSat is bounded above by $\zeta(\kappa,n)$.
\end{lemma}
\begin{proof}[Proof of Lemma \ref{lem:bound_from_poa}]
From Theorem 2 of \citet{christodoulou2011performance}, we have
\begin{equation}\label{bound_c0}
	C(\vec{f}^\kappa) \leq \zeta(\kappa,n)C(\vec{f}) \quad \forall \vec{f}\in \set{F}.
\end{equation}
Picking $\vec{f}=\vec{f}^0$ in \eqref{bound_c0}, we obtain the upper bound on $C(\vec{f}^\kappa)$. 
Inequalities involving $C(\vec{f}^*)$ are from the fact $C(\vec{f}^*)\leq C(\vec{f})$ for all $\vec{f}\in \set{F}$.
\end{proof}

The bound in Lemma \ref{lem:bound_from_poa} is not tight when $\kappa$ is small. For example, when $\kappa = 0, C(\vec{f^\kappa})=C(\vec{f^0})$.
Thus, $\frac{C(\vec{f^\kappa})}{C(\vec{f^0})}=1$. On the other hand,  (19) yields the following:
\begin{align*}
\frac{C(\vec{f^\kappa})}{C(\vec{f^0})} \le \bigg(1-\frac{n}{(n+1)^{\frac{n+1}{n}}}\bigg)^{-1}
\end{align*}
where the expression on the right is strictly larger than one. For example, the expression
reduces to $\frac{4}{3}$ when $n=1$ and approaches infinity when $n$ is large.

In Lemma 3 of \citet{christodoulou2011performance}, the existence of a network instance with $C(\vec{f}^\kappa) = (1+\kappa)^{n+1} C(\vec{f}^*)$ is shown for $\kappa \geq (n+1)^{1/n} - 1$ via a circular network example presented in Figure \ref{fig:circular}.
\begin{figure}\centering
\begin{tikzpicture}
\tiny
\filldraw (-20mm,0) circle (2pt) (20mm,0) circle (2pt);
\filldraw (0,-20mm) circle (2pt) (0,20mm) circle (2pt);
\filldraw (17.3205mm,10mm) circle (2pt) (10mm,17.3205mm) circle (2pt);
\filldraw (-17.3205mm,10mm) circle (2pt) (-10mm,17.3205mm) circle (2pt);
\filldraw (17.3205mm,-10mm) circle (2pt) (10mm,-17.3205mm) circle (2pt);
\filldraw (-17.3205mm,-10mm) circle (2pt) (-10mm,-17.3205mm) circle (2pt);
\draw (0,0) circle (20mm);
\draw [dashed,->](17.736539mm,3.125667mm) arc (10:110:18mm);
\draw [dashed,->](17.736539mm,-3.125667mm) arc (-10:-230:18mm);
\node[] at (16mm,0) {$i$};
\node[] at (-8mm,14mm) {$i+m$};
\node[] at (0mm,23mm) {$m+k$};
\node[] at (18mm,13.3205mm) {$2$};
\node[] at (10mm,20.3205mm) {$1$};
\end{tikzpicture}
\caption{Circular Network of \citet{christodoulou2011performance}}
\label{fig:circular}
\end{figure}
The circular network includes $m+l$ nodes where positive integers $m$ and $l$ are chosen so that $\frac{m}{l}=1+\kappa$.
All nodes lie in a circle and each node $i$ is adjacent to two neighboring nodes via arc $(i, i+1)$ and $(i-1,i)$.
The arc cost function for arc $a$ is $t_a(v_a) = (v_a)^n$, where $v_a$ is the total arc flow in arc $a$.
There are $m+l$ OD pairs $(i, i+m)$ for $i=1,2,..,m+l$, with unit demand from node $i$ to node $i+m$
 (indices are taken cyclically). 
Note that the circular network can be easily converted to a directed network by replacing each undirected arc with two directed arcs with opposite direction. 
The cost associated with arc $a$ will be $t_a(v_a,\hat{v}_a) = (v_a+\hat{v}_a)^n$ in this case where $\hat{v}_a$ is the flow in the arc with opposite direction. 

For each OD pair $(i, i+m)$, there are two paths, clockwise and
counterclockwise. The former contains $m$ arcs and the latter has only $l$. Note
that our choice requires that $\frac{m}{l}=(1+\kappa)$ where $\kappa \ge 0$, i.e., $m \ge l$.
Consider the all-or-nothing strategy that sends the unit demand for every OD pair
along only one path, clockwise or counterclockwise. Using the clockwise
strategy, each arc has $m$ units of flows and costs $m^n$. Thus, the clockwise path
costs $m \cdot m^n$, while the cost of the counterclockwise one is $l \cdot m^n$. Because
$\frac{m \cdot m^n}{l \cdot m^n}=\frac{m}{l} \ge 1$, the clockwise strategy is not in a user equilibrium unless $m=l$.
For the counterclockwise strategy, each arc has $l$ units of flows instead and costs
$l^n$. Similarly, the clockwise and counterclockwise path cost $m \cdot l^n$ and $l \cdot l^n$,
respectively. Using the same reasoning as before, flow-bearing (or the
counterclockwise) paths are less expensive than or the same as the unused ones.
Thus, the counterclockwise strategy is in user equilibrium and yields
$(m+l) \cdot l^{n+1}$ as the total travel cost.
Consider MSatUE. Because it is PRUE, the counterclockwise strategy is
automatically in MSatUE. But, the clockwise one is also in MSatUE because the
cost of flow-bearing (or the clockwise) path is exactly $(1+\kappa)$ time the cost of
the shortest path, i.e., the counterclockwise one. Additionally, the total travel
cost of the clockwise strategy is $(m+l) \cdot m^{n+1}$. Then, the PoSat of this circular
network is
\begin{align*}
\frac{(m+l) \cdot m^{n+1}}{(m+l) \cdot l^{n+1}}=\Big(\frac{m}{l}\Big)^{n+1}=(1+\kappa)^{n+1}.
\end{align*}
Thus, the bound is tight. Now consider the system problem for the circular network. The problem objective is
to minimize $\sum_{a \in \Ac} v_a t_a(\vec{v})$. Using the counterclockwise strategy, the partial derivative
of the objective function with respect to $v_a$ is $(n+1)l^n$. Then, switching to the
clockwise path would increase the total travel cost by $(m - l)(n+1)l^n \ge 0$  per unit
flow. Thus, the counterclockwise strategy is system optimal because switching to the
unused path does not lead to a reduction in the total travel cost.
Observe that the $\UEPEX$ model can capture the satisficing behavior of network users in the circular network adequately.
Let us consider $\vec{\lambda}$ as follow as (indices are taken cyclically):
\begin{align*}
 \lambda^{w_i}_a= 
			\begin{cases} 
				\frac{1}{1+\kappa} & \text{for } a \in \{(j,j+1) : j=i, i+1, ..., i+m-1\}\\
				1 &  \text{for } a \in \{(j-1,j) : j=i, i-1, ..., i-l+1\}
			\end{cases}
\end{align*}
Under the above ${\vec{\lambda}}$, if all network users choose the clockwise path, the flow in each arc will be equal to $m$, the path cost for each OD pair will be $\frac{1}{1+k} m^{n+1}=lm^n$, which is equal to the cost of the alternative path, and thus the clockwise path is a solution to $\UEPEX$ and $\PoSat_\kappa = (1+\kappa)^{n+1}$.
In Section \ref{sec:numerical_PoSat}, we will compute the PoSat numerically for these examples to confirm that $\UEPEX$ is a useful model to find $\PoSat_\kappa$.

\begin{figure}\centering
\begin{subfigure}[t]{0.49\textwidth}
	\resizebox{\textwidth}{!}{\input{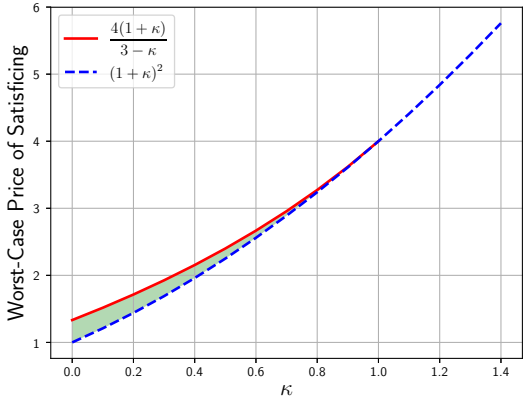}}
	\caption{$n=1$}
	\label{fig:ex1}
\end{subfigure}
\begin{subfigure}[t]{0.49\textwidth}
	\resizebox{1.03\textwidth}{!}{\input{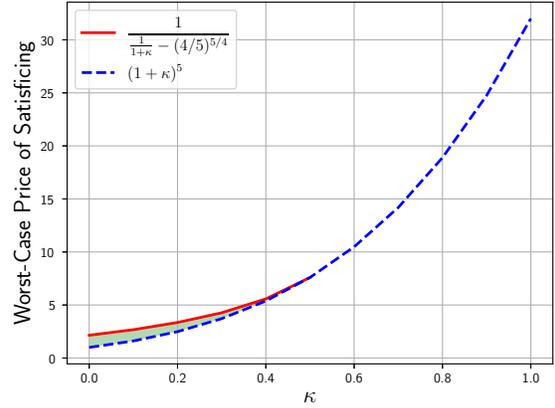}}
	\caption{$n=4$}
	\label{fig:ex2}
	\end{subfigure}
\caption{The worst-case price of satisficing for $n=1$ and $n=4$.
Note that when $n=1$, the right-hand-side of \eqref{pos} becomes $\frac{4(1+\kappa)}{3-\kappa}$.
For $n=1$, when $\kappa\geq 1$, we know for sure that the worst-case price of satisficing is exactly the dotted line. 
When $\kappa\leq 1$, the worst case falls in the shaded interval between the solid line and dotted line. 
When $n=4$, it is similar.}
\label{fig:pos}
\end{figure}

By Lemma \ref{lem:bound_from_poa}, when travel time functions are polynomials of degree $n$ with nonnegative coefficients, the PoSat is bounded as follows:
\begin{equation} \label{pos}
\PoSat_\kappa(G, \vec{Q}, \vec{t}) \leq
\begin{cases}
(1+\kappa)^{n+1} & \text{ if } \kappa \geq (n+1)^{1/n} - 1 ,\\
\bigg(\frac{1}{1+\kappa} - \frac{n}{(n+1)^{(n+1)/n}} \bigg)^{-1} & \text{ if } 0 \leq \kappa \leq (n+1)^{1/n} - 1.
\end{cases}		
\end{equation}
for all $(G, \vec{Q}, \vec{t}) \in \Omega(n) $ and by the circular network example in Figure \ref{fig:circular}, we know that there indeed exists a network instance $(G, \vec{Q}, \vec{t})\in\Omega(n)$ such that $\PoSat_\kappa(G, \vec{Q}, \vec{t}) = (1+\kappa)^{n+1} $ for all $\kappa \geq 0 $.
Therefore when $\kappa \geq (n+1)^{1/n} - 1$, the bound in \eqref{pos} is tight. 
Figure \ref{fig:ex1} shows the bounds in \eqref{pos} when travel time
functions are linear or when $n=1$.
For smaller $\kappa$ values, the worst-case PoSat falls in the shaded interval, while for larger $\kappa$ values, it is exactly $(1+\kappa)^2$.
Figure \ref{fig:ex2} shows the same bounds when
$n=4$ instead. 
When $\kappa$ is zero, we have $\vec{f}^\kappa=\vec{f}^0$; hence, we must have the PoSat approach to 1. 
With this observation, we naturally ask a question: Does $(1+\kappa)^{n+1}$ provide a tight bound on $\PoSat_\kappa$ for all $\kappa\geq 0$?
We present partial answers to this question in the following sections.

\subsection{Increased Travel Demands and Travel Time Functions}

We first define new sets of flow vectors. 
When the travel demand $Q_w$ for each $w\in\Wc$ is multiplied by the factor $1+\kappa$, we define
\begin{align*}
	\set{F}_{1+\kappa} &= \bigg\{ \vec{f} :  \sum_{p\in\Pc_w} f_p = (1+\kappa) Q_w \quad \forall w\in\Wc, \qquad f_p \geq 0 \quad \forall p\in\Pc \bigg \}, \\
	\set{V}_{1+\kappa} &= \bigg\{ \vec{v} : 	v_a = \sum_{p\in\Pc} \delta^p_a f_p \quad \forall a\in\Ac, \qquad \vec{f}\in \set{F}_{1+\kappa}  \bigg\}, \\
	\set{X}_{1+\kappa} &= \bigg\{ \vec{x} : x^w_a = \sum_{p\in\Pc_w} \delta^p_a f_p \quad \forall a\in\Ac,w\in\Wc \qquad \vec{f}\in \set{F}_{1+\kappa} \bigg\}.
\end{align*}
The above three sets can equivalently be written as follows:
\begin{align*}
	\set{F}_{1+\kappa} &= \{ (1+\kappa)\vec{f} : \vec{f}\in \set{F} \},  \\
	\set{V}_{1+\kappa} &= \{ (1+\kappa)\vec{v} : \vec{v}\in \set{V} \},  \\
	\set{X}_{1+\kappa} &= \{ (1+\kappa)\vec{x} : \vec{x}\in \set{X} \}.
\end{align*}
We will use `hat' for flow vectors in these sets, for example, $\vec{\hat{f}}^\kappa\in \set{F}_{1+\kappa}$, while without hat in the original sets as in $\vec{f}^\kappa\in \set{F}$.

We consider cases when the travel time functions $t_a(\cdot)$ are polynomials of order $n$, in particular, the following form of \emph{asymmetric} arc travel time function for each $a\in\Ac$:
\begin{align}
	t_a(\vec{v}) 
	   &= \sum_{m=0}^n {b}_{am} \bigg( \sum_{e\in\Ac} d_{aem} v_e \bigg)^m  \nonumber \\
	   &= \sum_{m=0}^n {b}_{am} \Big( \vec{d}_{am}^\top \vec{v} \Big)^m  \label{asymmetric_time}
\end{align}
for some constants ${b}_{am}$ for $m=0,1,...,n$ and $d_{aem}$ for $e\in\Ac$ and $m=0,1,...,n$.
Note that we use the vector form $\vec{d}_{am} = (d_{aem} : e\in\Ac)$.
The travel time function \eqref{asymmetric_time} is a general form of the travel time functions considered in the traffic equilibrium literature \citep{meng2014asymmetric, panicucci2007path}.
If $\vec{d}_{am}$ is a unit vector such that $d_{aem}$ is 1 if $a=e$ and 0 otherwise, we have a separable polynomial arc travel time function that has been used in the literature popularly \citep{christodoulou2011performance,roughgarden2002bad}:
\begin{equation} \label{separable}
	t_a(v_a) = \sum_{m=0}^n {b}_{am} ( v_a )^m = {b}_{a0} + {b}_{a1} v_a + {b}_{a2} (v_a)^2 + \cdots + {b}_{an} (v_a)^n.
\end{equation}

\begin{lemma} \label{lem:simple_bound}
With the polynomial travel time function \eqref{asymmetric_time}, for any $\vec{f} \in\set{F}$, we have
\begin{equation}
	C((1+\kappa)\vec{f}) \leq (1+\kappa)^{n+1} C(\vec{f})
\end{equation}
for all $\kappa \geq 0$ and $n\geq 0$.
\end{lemma}
\begin{proof}[Proof of Lemma \ref{lem:simple_bound}]
By simple comparison, we can show
\begin{align*}
	C((1+\kappa)\vec{f}) = Z((1+\kappa)\vec{v})
	& = \sum_{a\in\Ac}  \bigg( \sum_{m=0}^n {b}_{am} \Big( (1+\kappa) \vec{d}_{am}^\top \vec{v} \Big)^m  \bigg)  (1+\kappa) v_a \\
	& \leq (1+\kappa)^{n+1} \sum_{a \in \Ac} \bigg( \sum_{m=0}^n {b}_{am} \Big( \vec{d}_{am}^\top \vec{v} \Big)^m \bigg) v_a \\	
	& = (1+\kappa)^{n+1} Z(\vec{v}) \\
	& = (1+\kappa)^{n+1} C(\vec{f})
\end{align*}
where $\vec{v}$ is the arc flow vector from $\vec{f} \mapsto \vec{v}$.
\end{proof}

\subsection{Cases with Separable, Monomial Arc Travel Time Functions}

As a simple case, we consider separable, monomial functions of degree $n$ for arc travel time of the following form:
\begin{equation} \label{monomial}
	t_a(v_a) = b_{a} (v_a)^n 
\end{equation}
with a positive scalar $b_{a}$ for each $a\in\Ac$ and nonnegative constant $n$.

It is well known \citep{beckmann1956studies} that $\vec{v}^0\in\set{F}$ is a user equilibrium flow, if and only if it minimizes the following potential function
\[
	\Phi(\vec{v}) = \sum_{a\in\Ac} \int_0^{v_a} t_a(u) \du = \sum_{a\in\Ac} \frac{b_a}{n+1}(v_a)^{n+1}
\]
when the arc travel time functions are separable, so that the integral is well defined.
Similarly, $\vec{v}^\kappa\in\set{V}$ is a $\kappa$-MSatUE flow, if it is a solution to $\UEPEV$, or equivalently, if it minimizes the following potential function \citep{christodoulou2011performance}
\[
	\Psi(\vec{v};\vec{\lambda}) = \sum_{a\in\Ac} \int_0^{v_a} \lambda_a t_a(u) \du = \sum_{a\in\Ac} \frac{\lambda_a b_a}{n+1}(v_a)^{n+1}
\]
for some $\lambda_a \in [\frac{1}{1+\kappa},1]$ for each $a\in\Ac$.

When travel time functions are separable, we can show the following result \citep{englert2010sensitivity, takalloo_sensitivity}:

\begin{lemma} \label{lem:englert}
When the arc travel time functions are in the form of \eqref{separable}, 
let $\vec{f}^0\in\set{F}$ and $\vec{\hat{f}}^0\in\set{F}_{1+\kappa}$ be the PRUE flows with the corresponding travel demands.
We can show
\begin{equation}
C(\vec{\hat{f}}^0) \leq (1+\kappa)^{n+1} C(\vec{f}^0) 
\end{equation}
for all $\kappa\geq 0$ and $n\geq 0$.
\end{lemma}

Although \citet{englert2010sensitivity} consider cases with a single OD pair only with interest in the changes in the path travel time, the same technique can be used to prove Lemma \ref{lem:englert} for cases with multiple OD pairs. 
For completeness, we include the proof to Lemma \ref{lem:englert} in the appendix.

Using Lemma \ref{lem:englert}, we show that a solution to $\UEPEV$ is an MSatUE flow.

\begin{theorem} \label{thm:monomial}
When the arc travel time functions are of the form \eqref{monomial},
let $\bar{\vec{v}}\in\set{V}$ be a solution to $\UEPEV$ and $\bar{\vec{f}}\in \set{F}$ is the any corresponding path flow with $\bar{\vec{v}}\anymapsto\bar{\vec{f}}$.
We let $\vec{\hat{f}}^0\in \set{F}_{1+\kappa}$ be the PRUE flows.
Then we have
$C(\vec{f}^\kappa) \leq C(\vec{\hat{f}}^0),$
and consequently $C(\vec{f}^\kappa) \leq (1+\kappa)^{n+1} C(\vec{f}^0)$ for all $\kappa \geq 0$.
\end{theorem}
\begin{proof}[Proof of Theorem \ref{thm:monomial}]

Since $\hat{\vec{v}}^0\in\set{F}_{1+\kappa}$ is an user equilibrium flow that minimizes $\Phi(\vec{\cdot})$, we have
\[
\Phi(\hat{\vec{v}}^0) \le \Phi \big( (1+\kappa)\bar{\vec{v}} \big),
\]
which implies
\begin{equation}\label{A}
\sum_{a \in \Ac} \frac{b_{a}(\hat{v}^0_a)^{n+1}}{n+1} 
\le \sum_{a \in \Ac}\frac{b_{a} \big( (1+\kappa)\bar{v}_a \big)^{n+1}}{n+1}
= (1+\kappa)^{n+1} \sum_{a \in \Ac} \frac{b_{a}(\bar{v}_a)^{n+1}}{n+1}.
\end{equation}

Since $\bar{\vec{v}}\in\set{V}$ is a solution to $\UEPEV$, we have
\[
\Psi(\bar{\vec{v}}; \vec{\lambda}) \le \Psi \Big(\frac{\hat{\vec{v}}^0}{1+\kappa}; \vec{\lambda} \Big),
\] 
for some $\vec{\lambda}$.
Therefore, we have
\[
\sum_{a \in \Ac} \frac{\lambda_a b_{a}(\bar{v}_a)^{n+1}}{n+1} 
\le \sum_{a \in \Ac}\frac{\lambda_a b_{a}(v^0_a)^{n+1}}{(n+1)(1+\kappa)^{n+1}}
= \frac{1}{(1+\kappa)^{n+1}} \sum_{a \in \Ac} \frac{\lambda_a b_{a}(\hat{v}^0_a)^{n+1}}{n+1}.
\]
Since $\lambda_a \in [\frac{1}{1+\kappa},1]$, we obtain
\[
\frac{1}{1+\kappa}\sum_{a \in \Ac} \frac{b_{a}(\bar{v}_a)^{n+1}}{n+1} 
\le \frac{1}{(1+\kappa)^{n+1}}\sum_{a \in \Ac} \frac{b_{a}(\hat{v}^0_a)^{n+1}}{n+1}
\]
which implies
\begin{equation}\label{B}
(1+\kappa)^{n}\sum_{a \in \Ac} \frac{b_{a}(\bar{v}_a)^{n+1}}{n+1} \le 
\sum_{a \in \Ac} \frac{b_{a}(\hat{v}^0_a)^{n+1}}{n+1}
\end{equation}

Let us assume that $C(\bar{\vec{f}}) > C(\hat{\vec{f}}^0)$,
which is equivalent to 
\begin{equation}\label{C}
\sum_{a \in \Ac} b_{a}(\hat{v}^0_a)^{n+1} <
\sum_{a \in \Ac} b_{a}(\bar{v}_a)^{n+1} 
\end{equation}

From $A\times\eqref{A} + B\times\eqref{B} + C\times\eqref{C}$ for any positive constants $A$, $B$ and $C$, we obtain
\begin{equation}\label{cond}
\theta_1 \sum_{a \in \Ac}\frac{b_{a}(\hat{v}^0_a)^{n+1}}{n+1} 
< \theta_2 \sum_{a \in \Ac}\frac{b_{a}(\bar{v}_a)^{n+1}}{n+1}
\end{equation}
where
\begin{align*}
\theta_1 &= A-B+C(n+1) \\
\theta_2 &= A(1+\kappa)^{n+1}- B(1+\kappa)^n+C(n+1).
\end{align*}
In particular, consider $A$, $B$ and $C$ as follows:
\begin{align*}
A &= (n+1) \big( (1+\kappa)^n - 1 \big) \\
B &= (n+1) \big( (1+\kappa)^n - 1 \big) + (n+1)\kappa(1+\kappa)^{n+1} \\
C &=\kappa(1+\kappa)^{n+1}
\end{align*}
We observe that $A$, $B$ and $C$ are all positive and $\theta_1=0$.
We also see that 
\[
	\theta_2 = -(n+1) \kappa^2 (1+\kappa)^n \big( (1+\kappa)^{n+1} - 1 \big) \leq 0
\]
for all $\kappa \geq 0$ and $n\geq 0$, which leads to a contradiction. 
Therefore, we have 
\[
C(\bar{\vec{f}}) \leq C(\vec{\hat{f}}^0) \leq (1+\kappa)^{n+1} C(\vec{f}^0),
\]
where the last inequality is from Lemma \ref{lem:englert}.
This completes the proof.
\end{proof}

Note that the bound obtained in Theorem \ref{thm:monomial} relies on the sufficient condition, not a necessary condition.
Therefore, the result is not applicable to all MSatUE flows, although it provides a useful bound in the framework of UE-PE models.

\subsection{Cases with Separable Arc Travel Time Functions}

We consider general polynomial, separable arc travel functions in the form of \eqref{separable}.

\begin{theorem} \label{thm:compare_separable}
Suppose that the arc travel time functions are in the form of \eqref{separable}.
Let $\vec{f}^\kappa\in \set{F}$ be any $\kappa$-MSatUE and $\vec{\hat{f}}^0\in \set{F}_{1+\kappa}$ be the PRUE flow. 
Suppose that $\kappa \geq 0$ is sufficiently small, in particular, so that
\begin{equation} \label{new_strange_condition}
	\sum_{p\in\Pc} [c_p(\vec{\hat{f}}^0) - c_p(\vec{f}^\kappa)] (\hat{f}^0_p - f^\kappa_p) 
	\geq 
  \kappa \sum_{p\in\Pc} c_p(\vec{f}^\kappa) \Big| \hat{f}^0_p - f^\kappa_p \Big| .	
\end{equation}
Then we have
$	C(\vec{f}^\kappa) \leq C(\vec{\hat{f}}^0). $
Consequently 
$	C(\vec{f}^\kappa) \leq (1+\kappa)^{n+1} C(\vec{f}^0),$
and 
\[
	\sup_{(G, \vec{Q}, \vec{t}) \in \Omega(n)} \PoSat_\kappa(G, \vec{Q}, \vec{t}) = (1+\kappa)^{n+1}.
\]
\end{theorem}
\begin{proof}[Proof of Theorem \ref{thm:compare_separable}]
By slightly modifying the proof of Theorem \ref{thm:compare_asymmetric}, we can show $C(\vec{f}^\kappa) \leq C(\vec{\hat{f}}^0).$
By Lemmas \ref{lem:simple_bound} and \ref{lem:englert}, we complete the proof.
\end{proof}

Theorem \ref{thm:compare_separable} depends on condition \eqref{new_strange_condition} and a similar condition appears in general asymmetric cases as in Theorem \ref{thm:compare_asymmetric}.
We discuss this condition in Section \ref{sec:illustrative}.

\subsection{General Cases with Asymmetric Arc Travel Time Functions}

We consider asymmetric arc travel time functions \eqref{asymmetric_time}, in which case Lemma \ref{lem:englert} is not applicable.
We first observe that the multiple of a PRUE flow, $(1+\kappa) \vec{f}^0$, provides a satisficing solution to the traffic equilibrium problem with the increased travel demand.

\begin{lemma} \label{lem:UE_MSatUE}
Suppose $t_a(\cdot)$ are polynomials of order $n$ as defined \eqref{asymmetric_time}.
If $\vec{f}^0\in \set{F}$ is a PRUE flow, then $(1+\kappa)\vec{f}^0$ is a $\sigma$-MSatUE flow with $\sigma = (1+\kappa)^n-1$ in $\set{F}_{1+\kappa}$. 
When $n=1$, we have $\sigma=\kappa$.
\end{lemma}
\begin{proof}
Let $\vec{\bar{f}} = (1+\kappa)\vec{f}^0$, and $\vec{\bar{v}}=(1+\kappa)\vec{v}^0$ for the corresponding arc flow vectors.
If the condition 
\begin{equation}\label{MSatPP}
	\sum_{a\in\Ac} \bigg( \sum_{m=0}^n \lambda_{am} {b}_{am} \Big( \vec{d}_{am}^\top \bar{\vec{v}} \Big)^m  \bigg) ( v'_a - \bar{v}_a ) \geq 0 \qquad \forall \vec{v}' \in \set{V}_{1+\kappa}
\end{equation}
holds for some constants $\lambda_{am} \in [\frac{1}{1+\sigma},1]$ for $m=0,1,...,n$ and $a\in\Ac$,
then we can find $\lambda_a \in [\frac{1}{1+\sigma},1]$ such that
\[
	\lambda_a \sum_{m=0}^n {b}_{am} \Big( \vec{d}_{am}^\top \bar{\vec{v}} \Big)^m   = \sum_{m=0}^n \lambda_{am} {b}_{am} \Big( \vec{d}_{am}^\top \bar{\vec{v}} \Big)^m  
\]
for all $a\in\Ac$;
consequently, by Lemmas \ref{lem:ue-pe-x} and \ref{lem:ue-pe-v}, $\bar{\vec{f}}$ is a $\sigma$-MSatUE flow in $\set{F}_{1+\kappa}$.

Since $\vec{v}^0$ is PRUE for $\set{V}$, we know that
\[
	\sum_{a\in\Ac} \bigg( \sum_{m=0}^n {b}_{am} \Big( \vec{d}_{am}^\top \vec{v}^0 \Big)^m  \bigg) ( v_a - v^0_a ) \geq 0 \qquad \forall \vec{v}\in \set{V}.
\]
Therefore
\[
	\sum_{a\in\Ac} \bigg( \sum_{m=0}^n \frac{1}{(1+\kappa)^m} {b}_{am} \Big( (1+\kappa) \vec{d}_{am}^\top \vec{v}^0 \Big)^m  \bigg) ( (1+\kappa) v_a - (1+\kappa) v^0_a ) \geq 0 \qquad \forall \vec{v}\in \set{V}.
\]
Letting for all $a\in\Ac$
\begin{align*}
	\lambda_{am} & = \frac{1}{(1+\kappa)^m}, \qquad m=0,1,...,n \\
	\bar{v}_a &= (1+\kappa) v^0_a, \\
	v'_a &= (1+\kappa) v_a, 
\end{align*}
we observe that $\lambda_{am} \in [\frac{1}{1+\sigma}, 1]$ and we obtain \eqref{MSatPP}; hence proof.
\end{proof}

By introducing an additional condition, we compare MSatUE flows with the proportional travel demand increase, and obtain the worst-case bound of PoSat.

\begin{theorem} \label{thm:compare_asymmetric}
Let $\vec{f}^\kappa\in \set{F}$ be any $\kappa$-MSatUE and $\vec{\hat{f}}^\sigma\in \set{F}_{1+\kappa}$ be any $\sigma$-MSatUE flows with the corresponding travel demands, when $\sigma = (1+\kappa)^n-1$. 
Suppose that $\kappa \geq 0$ is sufficiently small, in particular, so that
\begin{equation} \label{strange_condition}
	\sum_{p\in\Pc} [c_p(\vec{\hat{f}}^\sigma) - c_p(\vec{f}^\kappa)] (\hat{f}^\sigma_p - f^\kappa_p) 
	\geq 
  \sigma \sum_{p\in\Pc} \max\{ c_p(\vec{\hat{f}}^\sigma), c_p(\vec{f}^\kappa) \} \Big| \hat{f}^\sigma_p - f^\kappa_p \Big| .	
\end{equation}
Then we have
$
	C(\vec{f}^\kappa) \leq C(\vec{\hat{f}}^\sigma).
$
Consequently 
$	C(\vec{f}^\kappa) \leq (1+\kappa)^{n+1} C(\vec{f}^0),$
and 
\[
\sup_{(G, \vec{Q}, \vec{t}) \in \Omega(n)} \PoSat_\kappa(G, \vec{Q}, \vec{t}) = (1+\kappa)^{n+1}.
\]
\end{theorem}
\begin{proof}[Proof of Theorem \ref{thm:compare_asymmetric}]
We decompose $\Pc_w$ for each OD pair $w$ into the following four subsets:
\begin{align*}
\Pc^1_w &= \{ p\in\Pc_w : \hat{f}^\sigma_p > 0, \ f^\kappa_p > 0, \ \hat{f}^\sigma_p - f^\kappa_p \geq 0 \}, \\
\Pc^2_w &= \{ p\in\Pc_w : \hat{f}^\sigma_p > 0, \ f^\kappa_p > 0, \ \hat{f}^\sigma_p - f^\kappa_p < 0 \}, \\
\Pc^3_w &= \{ p\in\Pc_w : \hat{f}^\sigma_p > 0, \ f^\kappa_p = 0 \}, \\
\Pc^4_w &= \{ p\in\Pc_w : \hat{f}^\sigma_p = 0, \ f^\kappa_p > 0 \}.
\end{align*}
We ignore cases with $\hat{f}^\sigma_p = 0$ and $f^\kappa_p = 0$.
Note that $\hat{f}^\sigma_p - f^\kappa_p > 0$ for $p\in\Pc^3_w$ and $\hat{f}^\sigma_p - f^\kappa_p < 0$ for $p\in\Pc^4_w$.
From the definition of MSatUE flows, we have
\begin{align*}
\hat{f}^\sigma_p > 0 &\implies c_p(\vec{\hat{f}}^\sigma) \leq (1+\sigma) \mu_w(\vec{\hat{f}}^\sigma), \\
f^\kappa_p > 0       &\implies c_p(\vec{f}^\kappa) \leq (1+\kappa) \mu_w(\vec{f}^\kappa),
\end{align*}
for all $p\in\Pc_w, w\in\Wc$.
In addition, $\mu_w(\vec{\hat{f}}^\sigma) \leq c_p(\vec{\hat{f}}^\sigma)$ and $\mu_w(\vec{f}^\kappa) \leq c_p(\vec{f}^\kappa)$ for all $p\in\Pc$ by definition.
Therefore, we have
\begin{align*}
&\sum_{p\in\Pc} [c_p(\vec{\hat{f}}^\sigma) - c_p(\vec{f}^\kappa)] (\hat{f}^\sigma_p - f^\kappa_p) \\
& \leq \sum_{w\in\Wc}
	\bigg\{ \sum_{p\in\Pc^1_w} \Big[ (1+\sigma) \mu_w(\vec{\hat{f}}^\sigma) - \mu_w(\vec{f}^\kappa) \Big] (\hat{f}^\sigma_p - f^\kappa_p) 
+ \sum_{p\in\Pc^2_w} \Big[ \mu_w(\vec{\hat{f}}^\sigma) - (1+\kappa) \mu_w(\vec{f}^\kappa) \Big] (\hat{f}^\sigma_p - f^\kappa_p)	\\
&\qquad + \sum_{p\in\Pc^3_w} \Big[ (1+\sigma) \mu_w(\vec{\hat{f}}^\sigma) - \mu_w(\vec{f}^\kappa) \Big] (\hat{f}^\sigma_p - f^\kappa_p)	
 + \sum_{p\in\Pc^4_w} \Big[ \mu_w(\vec{\hat{f}}^\sigma) - (1+\kappa) \mu_w(\vec{f}^\kappa) \Big] (\hat{f}^\sigma_p - f^\kappa_p)	\bigg\} \\
&= \sum_{w\in\Wc}
	\bigg\{ \sum_{p\in\Pc_w} \Big[ \mu_w(\vec{\hat{f}}^\sigma) - \mu_w(\vec{f}^\kappa) \Big] (\hat{f}^\sigma_p - f^\kappa_p) 
+ \sigma \sum_{p\in\Pc^1_w \cup \Pc^3_w} \mu_w(\vec{\hat{f}}^\sigma) (\hat{f}^\sigma_p - f^\kappa_p)	\\
&\qquad\qquad - \kappa \sum_{p\in\Pc^2_w \cup \Pc^4_w} \mu_w(\vec{f}^\kappa) (\hat{f}^\sigma_p - f^\kappa_p)	\bigg\} \\
&\leq \sum_{w\in\Wc}
	\bigg\{ \sum_{p\in\Pc_w} \Big[ \mu_w(\vec{\hat{f}}^\sigma) - \mu_w(\vec{f}^\kappa) \Big] (\hat{f}^\sigma_p - f^\kappa_p) 
+ \sigma \sum_{p\in\Pc_w} \max\{\mu_w(\vec{\hat{f}}^\sigma),\mu_w(\vec{f}^\kappa)\} \Big| \hat{f}^\sigma_p - f^\kappa_p \Big|	\bigg\} \\
&\leq \sum_{w\in\Wc} \sum_{p\in\Pc_w} \Big[ \mu_w(\vec{\hat{f}}^\sigma) - \mu_w(\vec{f}^\kappa) \Big] (\hat{f}^\sigma_p - f^\kappa_p) 
+ \sigma \sum_{p\in\Pc} \max\{ c_p(\vec{\hat{f}}^\sigma), c_p(\vec{f}^\kappa) \} \Big| \hat{f}^\sigma_p - f^\kappa_p \Big| .
\end{align*} 
From \eqref{strange_condition}, we obtain
\begin{align*}
0 
&\leq \sum_{w\in\Wc} \sum_{p\in\Pc} \Big[ \mu_w(\vec{\hat{f}}^\sigma) - \mu_w(\vec{f}^\kappa) \Big] (\hat{f}^\sigma_p - f^\kappa_p) \\
&= \sum_{w\in\Wc} \Big[ \mu_w(\vec{\hat{f}}^\sigma) - \mu_w(\vec{f}^\kappa) \Big] \bigg(\sum_{p\in\Pc} \hat{f}^\sigma_p - \sum_{p\in\Pc} f^\kappa_p\bigg) \\
&= \sum_{w\in\Wc} \Big[ \mu_w(\vec{\hat{f}}^\sigma) - \mu_w(\vec{f}^\kappa) \Big] (\hat{Q}_w - Q_w) \\
&= \kappa \sum_{w\in\Wc} \mu_w(\vec{\hat{f}}^\sigma) Q_w - \kappa \sum_{w\in\Wc} \mu_w(\vec{f}^\kappa) Q_w \\
&= \frac{\kappa}{1+\kappa} \sum_{w\in\Wc} \mu_w(\vec{\hat{f}}^\sigma) \hat{Q}_w - \kappa \sum_{w\in\Wc} \mu_w(\vec{f}^\kappa) Q_w \\
&\leq \frac{\kappa}{1+\kappa} \sum_{w\in\Wc}\sum_{p\in\Pc_w}  c_p(\vec{\hat{f}}^\sigma) \hat{f}^\sigma_p
      - \frac{\kappa}{1+\kappa} \sum_{w\in\Wc}\sum_{p\in\Pc_w}  c_p(\vec{f}^\kappa) f^\kappa_p \\
&= \frac{\kappa}{1+\kappa} C(\vec{\hat{f}}^\sigma) - \frac{\kappa}{1+\kappa} C(\vec{f}^\kappa).
\end{align*}
Lemmas \ref{lem:simple_bound} and \ref{lem:UE_MSatUE} complete the proof.
\end{proof}

Note that condition \eqref{strange_condition} is stronger than condition \eqref{new_strange_condition} for separable travel time functions.
This is natural, since we consider more general classes of travel time functions.

\begin{figure}\centering
\begin{subfigure}[b]{0.4\textwidth}
	\begin{tikzpicture}[>=stealth',shorten >=1pt,auto,node distance=2.8cm]
		\node[shape=circle, draw] (O) at (0,0) {O};
		\node[shape=circle, draw] (D) at (5,0) {D};
		\path[->] (O) edge [bend left=20] node {$t_1(v_1)=1$} (D);
		\path[->] (O) edge [bend right=20, below] node {$t_2(v_2)=1+v_2$} (D);
	\end{tikzpicture}
	\caption{Example 1}
\end{subfigure}
\begin{subfigure}[b]{0.4\textwidth}
	\begin{tikzpicture}[>=stealth',shorten >=1pt,auto,node distance=2.8cm]
		\node[shape=circle, draw] (O) at (0,0) {O};
		\node[shape=circle, draw] (D) at (5,0) {D};
		\path[->] (O) edge [bend left=20] node {$t_1(v_1)=v_1$} (D);
		\path[->] (O) edge [bend right=20, below] node {$t_2(v_2)=v_2$} (D);
	\end{tikzpicture}
	\caption{Example 2}	
\end{subfigure}
\caption{Examples where the travel demand is $Q$ from node O to node D.}
\label{fig:simple2}
\end{figure}
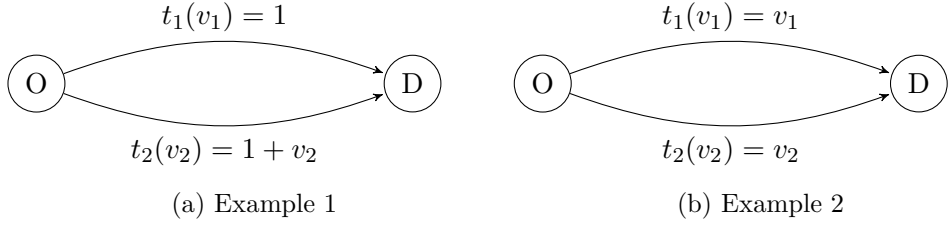

\subsection{Illustrative Examples} \label{sec:illustrative}

For the illustration purpose, we consider two examples in Figure \ref{fig:simple2} with linear travel time functions, where $n=1$.
In Example 1, the travel time function in the first arc is not increasing. 
We can verify that
\begin{align*}
\max C(\vec{f}^\kappa) = 
			\begin{cases} 
				Q + \kappa^2 & \text{ if } \kappa \leq Q, \qquad \text{ with } \vec{f}^\kappa = (Q-\kappa, \kappa) \\
				(1+Q)Q & \text{ if } \kappa \geq Q, \qquad \text{ with } \vec{f}^\kappa = (0, Q) 
			\end{cases}
\end{align*}
among all $\kappa$-MSatUE flows in $\set{F}$ and
\begin{align*}
C(\hat{\vec{f}}^0) = (1+\kappa) Q	\qquad\text{ with } \hat{\vec{f}}^0&= (1+\kappa)\vec{f}^0 = ((1+\kappa)Q, 0).
\end{align*}
among all $\kappa$-MSatUE flows in $\set{F}_{1+\kappa}$.
Comparing the two quantities, we observe $C(\vec{f}^\kappa) \leq C(\vec{\hat{f}}^0)$ in both cases.
To prove Theorem \ref{thm:compare_asymmetric}, condition \eqref{new_strange_condition} needs to hold only for these two flow vectors.
Regardless of the value of $\kappa$, however, it is impossible to satisfy condition \eqref{new_strange_condition}, although the worst-case PoSat bound $(1+\kappa)^{n+1}$ still holds for all $\kappa\geq 0$.
The price of satisficing is $1+\frac{\kappa^2}{Q}$ if $\kappa<Q$ and $1+Q$ if $\kappa\geq Q$ in this example, both of which are less than $(1+\kappa)^2$.

On the other hand, in Example 2, we have strictly monotone travel time functions in both arcs.
Similarly, we consider
\begin{align*}
	\max C(\vec{v}^\kappa) &= \frac{2+2\kappa+\kappa^2}{(2+\kappa)^2}Q &&\qquad\text{ with } \vec{f}^\kappa = \bigg( \frac{Q}{2+\kappa}, \frac{(1+\kappa)Q}{2+\kappa} \bigg ) \\
      C(\hat{\vec{v}}^0) &=\frac{(1+\kappa)^2}{2}Q &&\qquad\text{ with } \hat{\vec{f}}^0 = (1+\kappa)\vec{f}^0 = \bigg( \frac{(1+\kappa)Q}{2}, \frac{(1+\kappa)Q}{2} \bigg)
\end{align*}
and can verify that $C(\vec{f}^\kappa) \leq C(\vec{\hat{f}}^0)$ for all $\kappa \geq 0$.
In Example 2, we note that \eqref{strange_condition} holds for $\kappa \leq 0.206$.
In this example, we observe that the price of satisficing is $\frac{2(2+2\kappa+\kappa^2)}{(2+\kappa)^2}$, which is no greater than $(1+\kappa)^2$ for all $\kappa \geq 0$.

\subsection{Other Approaches}\label{otherapproach}

When there is a single origin and multiple destinations, i.e., a single common origin node, in the network, \citet{kleer2016impact} introduces the notion of the \emph{deviation ratio} that compares the system performances of the user equilibrium and the equilibrium with \emph{deviated} travel time functions $\tilde{t}_a(\cdot)$.
The notion of deviation may also be interpreted as perception in our definition.
In a special case, the deviation ratio is reduced to the price of risk aversion \citep{nikolova2015burden} that compares the performances of equilibria among risk-averse and risk-neutral network users.

\citet{kleer2016impact} define the (separable) deviated travel time functions with the following bounds:
\begin{equation} \label{dr_bound}
	t_a(v_a) + \alpha t_a(v_a) \leq \tilde{t}_a(v_a) \leq t_a(v_a) + \beta t_a(v_a)
\end{equation}
where $-1\leq \alpha \leq 0 \leq \beta$. 
The consideration of this deviated travel time function generalizes our UE-PE model where $\alpha = -\frac{\kappa}{1+\kappa}$ and $\beta = 0$.
\citet{kleer2016impact} show that the worst-case deviation ratio with \eqref{dr_bound} is bounded by
\[
	1 + \frac{\beta-\alpha}{1+\alpha} \bigg\lceil \frac{|\Nc| - 1}{2} \bigg\rceil Q.
\]
Therefore, we obtain the following theorem:

\begin{theorem}[\citealp{kleer2016impact}]\label{thm:deviation}
Consider a directed graph with a single common origin node with the total travel demand $Q$ and let $|\Nc|$ be the number of nodes. Then we have
\begin{equation} \label{deviation}
\frac{Z(\vec{v}^\kappa)}{Z(\vec{v}^0)} \leq 1 + \kappa \bigg\lceil \frac{|\Nc| - 1}{2} \bigg\rceil Q
\end{equation}
where $\vec{v}^\kappa$ is a solution to $\UEPEV$ in \eqref{ue-pe-v}.
\end{theorem}

Note that Theorem \ref{thm:deviation} only covers a subset of the entire MSatUE flows, as it is limited to the solutions $\UEPEV$ in \eqref{ue-pe-v} and is applicable to cases with a \emph{single} common origin.
When Theorem \ref{thm:deviation} is applied in the examples in Figure \ref{fig:simple2}, the bound \eqref{deviation} becomes $1+\kappa Q$.

\section{Numerical Bounds} \label{sec:numerical_PoSat}

To quantify PoSat in typical traffic networks and compare it with the analytical bound obtained in Theorem \ref{thm:compare_asymmetric}, we define the worst-case problem for the total system travel time under MSatUE as follows:
\begin{align} \label{mpcc0}
	\max_{\vec{v}^\kappa} &\quad  Z(\vec{v}^\kappa) = \sum_{a\in\Ac} z_a(v^\kappa_a) = \sum_{a\in\Ac} t_a(\vec{v}^\kappa) v_a^\kappa \\
	\text{subject to} & \quad\text{$\vec{v}^\kappa$ is an \underline{MSatUE} flow with $\kappa$} \nonumber
\end{align}
To quantify the benefit of satisficing, instead of maximizing, we can minimize the objective function \eqref{mpcc0}.
Since MSatUE involves path-based definition and formulation, \eqref{mpcc0} is numerically more challenging to solve. 
Instead, we replace MSatUE by $\UEPEX$. 
We know that the $\UEPEX$.  models provide a subset of MSatUE traffic flow patterns as seen in Lemmma \ref{lem:ue-pe-x}; hence by using $\UEPEX$ models, we will obtain suboptimal solutions to \eqref{mpcc0}. 

Using $\UEPEX$ in \eqref{ue-pe}, we formulate the worst-case problem as follows:
\begin{align} \label{mpec}
	\max_{\bar{\vec{v}}, \bar{\vec{x}}, \vec{\epsilon}} &\quad  Z(\bar{\vec{v}}) = \sum_{a\in\Ac} z_a(\bar{\vec{v}}) = \sum_{a\in\Ac} t_a(\bar{\vec{v}}) \bar{v}_a \\
	\text{subject to} 
	& \quad \sum_{a\in\Ac} \sum_{w\in\Wc} ( t_a(\bar{\vec{v}}) - \epsilon^w_a ) (x^w_a - \bar{x}^w_a) \geq 0 && \forall \vec{x}\in \set{X} \label{mpec1}  \\
	& \quad 	\bar{v}^\kappa_a = \sum_{w\in\Wc} \bar{x}^w_a  && \forall a\in\Ac \\
	& \quad \bar{x} \in \set{X} \\
	& \quad 0 \leq  \epsilon^w_a \leq \frac{\kappa}{1+\kappa} t_a(\bar{\vec{v}}) && \forall a\in\Ac  \label{mpec99}
\end{align}
Problem \eqref{mpec} is an instance of mathematical programs with equilibrium constraints (MPEC). 
We can replace the equilibrium condition \eqref{mpec1} by the following KKT conditions to create a single-level optimization problem:
\begin{align} 
	& \quad t_a(\bar{\vec{v}}) - \epsilon^w_a + \pi^w_i - \pi^w_j \geq 0  && \forall w\in\Wc, a\in\Ac \label{KKT00} \\
	& \quad \bar{x}^w_a ( t_a(\bar{\vec{v}}) - \epsilon^w_a + \pi^w_i - \pi^w_j ) = 0  && \forall w\in\Wc, a\in\Ac \label{KKT01}\\
	& \quad \sum_{a\in\Ac^+_i} \bar{x}^w_a - \sum_{a\in\Ac^-_i} \bar{x}^w_a = q^w_i && \forall w\in\Wc, i\in\Nc \label{KKT02} 
\end{align}
The resulting problem is a mathematical program with complementarity conditions (MPCC), which is nonlinear and nonconvex. 
Finding a global solution to MPCC problems is in general difficult, and \citet{kleer2016impact} has shown that solving the above MPCC optimally is NP-hard.
In order to solve this problem, we use an interior point method by utilizing the Ipopt nonlinear solver \citep{wachter2006implementation} with multiple starting solutions.

\subsection{Numerical Experiments} \label{sec:numerical_examples}
In this section we present some examples to compare the total travel times in $\MSatUE$ and PRUE numerically for both separable and asymmetric networks.
We approximate $\MSatUE$ by $\UEPEX$ and solve it by the Ipopt nonlinear solver, after reformulating \eqref{mpcc0} as a single-level optimization problem using KKT conditions.
We use the Julia Language and the JuMP package \citep{DunningHuchetteLubin2017} for modeling and interfacing with the Ipopt solver.

\begin{figure}
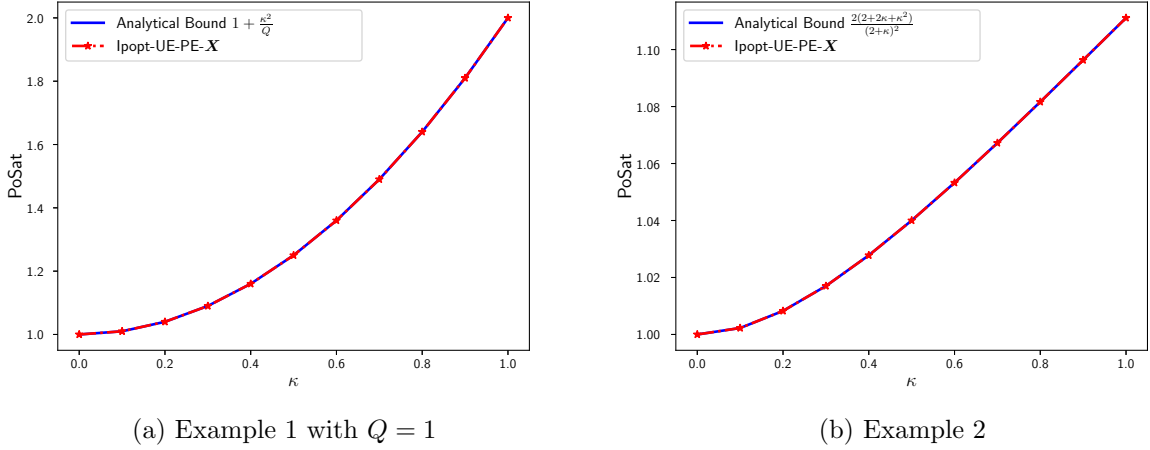
\centering
\begin{subfigure}[t]{0.49\textwidth}
	\resizebox{\textwidth}{!}{\input{graphics/example1.pgf}}
	\caption{Example 1 with $Q=1$}
	\label{fig:ex11}
\end{subfigure}
\begin{subfigure}[t]{0.49\textwidth}
	\resizebox{\textwidth}{!}{\input{graphics/example2.pgf}}
	\caption{Example 2}
	\label{fig:ex21}
	\end{subfigure}
\caption{$\PoSat$ for the simple networks in Figure \ref{fig:simple2}}
\label{fig:simple}
\end{figure}

\subsubsection{Simple Networks}\label{simple_example}
To test the validity and the strength of $\UEPEX$ model, we first consider Examples 1 and 2 in Figure \ref{fig:simple2}.
Figure \ref{fig:simple} compares the $\PoSat$ under $\UEPEX$ with the $\PoSat$ under $\MSatUE$, obtained in Section \ref{sec:illustrative}.
As Figure \ref{fig:simple} shows, the $\PoSat$ under $\UEPEX$ is equal to the the $\PoSat$ under $\MSatUE$ for both examples, which suggests that the $\UEPEX$ model is an effective model.

\subsubsection{Circular Network of \citet{christodoulou2011performance}}
We also compute the PoSat under $\UEPEX$ model for the circular network of \citet{christodoulou2011performance} presented in Figure \ref{fig:circular}.
For numerical experiments, we assign $m$ and $l$ to the smallest positive integers such that $\frac{m}{l}=(1+\kappa)^5$, and $\kappa \in \{0, 0.1, 0.2,...,1\}$. We also set $n=4$.
As it can be seen in Figure \ref{fig:circular_analytic}, we obtained identical results for circular network under $\UEPEX$, using the Ipopt solver, which shows that the $\UEPEX$ model can obtain the upper bound provided in Lemma \ref{lem:bound_from_poa}. 
\begin{figure}\centering 
	\resizebox{0.5\textwidth}{!}{\input{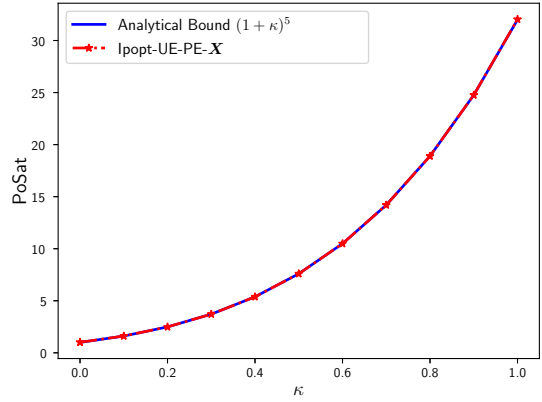}}
	\caption{$\PoSat$ for the circular network in Figure \ref{fig:circular}}
	\label{fig:circular_analytic}
\end{figure}

\subsubsection{Larger Networks}\label{sep_example}

We present some examples to compare the total travel times in $\MSatUE$ and $\PRUE$ numerically and compare the numerical worst-cases with the analytical bound given in Theorem \ref{thm:compare_separable} for larger networks with both separable and non-separable, asymmetric arc cost functions.
As \eqref{mpcc0} is a non-convex problem, the Ipopt solver can produce a local minimum at best.
To obtain a higher-quality local minimum, we solve the problem multiple times by using different initial solution.
For generating different initial solutions for the network with separable arc cost functions, we utilize $\UEPEV$ model. 
We generate initial $\vec{\lambda}$ randomly and use the Frank-Wolfe algorithm to obtain the corresponding $\vec{v}$ and $\vec{x}$.
For the network with non-separable cost function, we can use the fixed point method \citep{dafermos1980traffic} with a randomized $\vec{\lambda}$ to obtain an initial solution $\vec{x}$.
We randomly generate five initial starting points for each example and report the largest PoSat values.

We first consider the nine-node network presented in \citet{hearn1998solving}.
The nine-node network consists of 9 nodes and 18 arcs, and the travel time functions are polynomials of order $n=4$.
We also create an asymmetric variant of the nine-node network as shown in Figure \ref{fig6} in Appendix \ref{asynet}.
The asymmetric nine-node network has non-separable arc cost function in the form of \eqref{asyeq}. 
The comparison result is presented in Figure \ref{fig:nine-node}.
As Figure \ref{fig:nine-node} represents $\PoSat_\kappa$ increases with $\kappa$ for both symmetric and asymmetric nine-node network since $\PRUE$ total travel time is fixed with respect to $\kappa$, while the worst-case $\MSatUE$ total travel time increases as $\kappa$ increases.
Moreover, $\PoSat_\kappa$ is smaller for symmetric nine-node network compared to the asymmetric nine-node network for smaller $\kappa$ values (0.1 and 0.2), but it is greater for larger $\kappa$ values ($\kappa \ge 0.3$).
In general, the gap between $\PoSat_\kappa$ for symmetric nine-node network and asymmetric nine-node network is small.

Figure \ref{fig:nine-node-analytic} compares the numerical $\PoSat_\kappa$ with the analytical bound provided in Theorem \ref{thm:compare_separable} for $\MSatUE$ for the nine-node network. 
We observe that there is a large gap between the analytical and numerical bounds which increases with $\kappa$.
Although the analytical result certainly provides a valid bound, it is too large to be practically useful in realistic road networks. 
This indicates opportunities for empirical studies on the analytical bounds that depend on more network-specific information such as travel demands and travel time functions. 
The bound $(1+\kappa)^{n+1}$ in Theorem \ref{thm:compare_separable} is independent from such network-specific information.

\begin{figure}\centering
\begin{subfigure}[t]{0.49\textwidth}
	\resizebox{\textwidth}{!}{\input{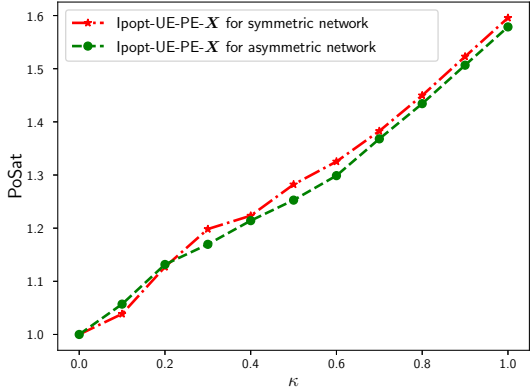}}
	\caption{Approximate $\PoSat$ (Ipopt-$\UEPEX$)}
	\label{fig:nine-node}
\end{subfigure}
\begin{subfigure}[t]{0.49\textwidth}
	\resizebox{\textwidth}{!}{\input{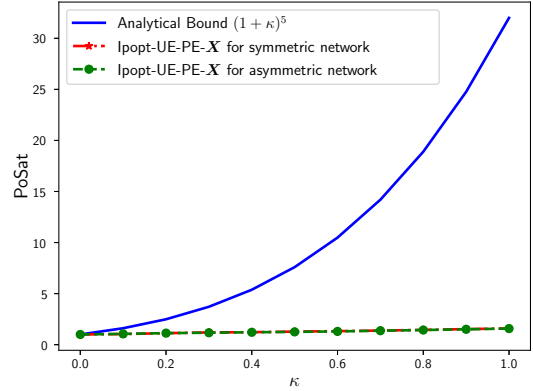}}
	\caption{Comparing analytical and numerical $\PoSat$}
	\label{fig:nine-node-analytic}
	\end{subfigure}
\caption{$\PoSat$ for nine-node network}
\label{fig:nine-node-figures}
\end{figure}

We also consider the Sioux Falls network presented in \citet{suwansirikul1987equilibrium}, which consists of 24 nodes, 76 arcs, and 576 OD pairs.
The arc travel cost function is the BPR function, which is a polynomial function with degree $n=4$.
We also consider an asymmetric variant of Sioux Falls network with arc cost function in the form of \eqref{asyeq}. 
As Figure \ref{fig:sioux-falls} represents, $\PoSat_\kappa$ increases with $\kappa$ for both symmetric and asymmetric Sioux Falls network, and it is greater compared to the nine-node network for both symmetric and asymmetric networks.
Furthermore,  $\PoSat_\kappa$ is greater for asymmetric Sioux Falls network compared with the symmetric Sioux Falls network for all positive $\kappa$ values,
and the gap between $\PoSat_\kappa$ for symmetric Sioux Falls network and $\PoSat_\kappa$ for asymmetric Sioux Falls network increases with $\kappa$.
Figure \ref{fig:sioux-falls-analytic} compares the numerical $\PoSat_\kappa$ with the analytical bound.
The gap between the analytical and the numerical bound is tighter compared to the nine-node network, but it is still considerable.

\begin{figure}
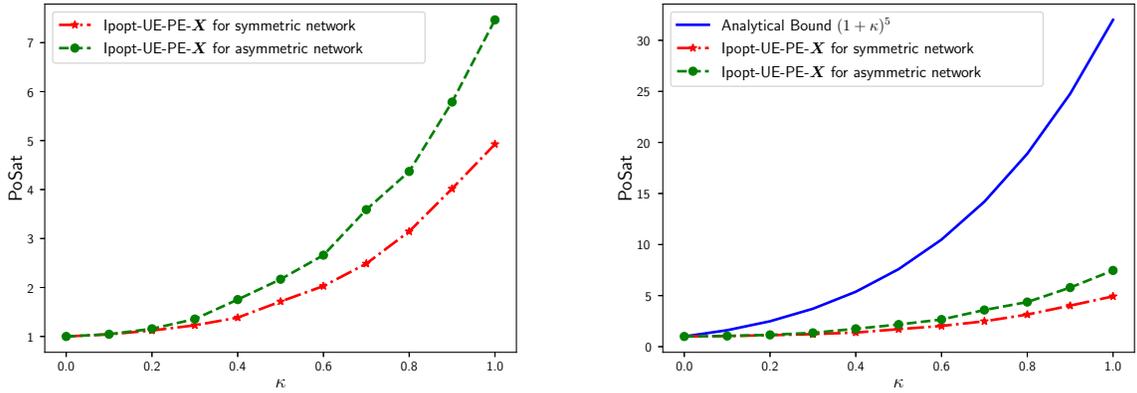
\centering
\begin{subfigure}[t]{0.49\textwidth}
	\resizebox{\textwidth}{!}{\input{graphics/asy_Posat1.pgf}}
	\caption{Approximate PoSat (Ipopt-$\UEPEX$)}
	\label{fig:sioux-falls}
\end{subfigure}
\begin{subfigure}[t]{0.49\textwidth}
	\resizebox{\textwidth}{!}{\input{graphics/asy_Posat2.pgf}}
	\caption{Comparing analytical and numerical PoSat}
	\label{fig:sioux-falls-analytic}
	\end{subfigure}
\caption{$\PoSat$ for Sioux Falls network}
\label{fig:sioux-falls-figures}
\end{figure}

\section{Concluding Remarks} \label{sec:conclusion}

When network users are satisficing decision makers, the resulting satisficing user equilibria may degrade the system performance, compared to the perfectly rational user equilibrium. 
To quantify how much the performance can deteriorate, this paper has quantified the worst-case analytical bound on the price of satisficing.
We also quantified the price of satisficing for several networks numerically and compare it to the analytical bound.

As we have seen in the numerical examples in this paper, there is a large gap between the worst-case analytical bound and the actual bound.
Clearly, this is a limitation of our approach.
In the literature of the price of anarchy have similar observations been reported \citep{o2016mechanisms, monnot2017bad, colini2017asymptotic}.
Likewise, the behavior of the price of satisficing in practice can be quite different from what we have observed in this paper.
Deriving empirical or network-specific bounds can be meaningful contributions as a future research direction.

We suggest additional potential future research directions. 
For the proposed analytical bound, our result is based on the condition \eqref{strange_condition}. 
By attempting to relax this condition, one may obtain a global bound for any value of $\kappa$.
In deriving the analytical bound, we utilized a novel technique comparing equilibrium patterns before and after the travel demand is increased; namely $\set{V}$ and $\set{V}_{1+\kappa}$.
Applying this technique in the context of the price of risk aversion and the deviation ratio would be an interesting research direction.

\section*{Acknowledgments}
This research was partially supported by the National Science Foundation under grant CMMI-1351357.
The authors are thankful for the anonymous reviewers whose constructive comments have greatly improved this manuscript.

\bibliographystyle{ormsv080-ck}
\bibliography{POS}

\appendix
\appendixpage

\section{Proof of Lemma \ref{lem:englert}}

\begin{proof}[Proof of Lemma \ref{lem:englert}]
This is a minor variant to the proof of \citet[Theorem 3]{englert2010sensitivity}.
Since $\vec{v}^0\in\set{F}$ and $\hat{\vec{v}}^0\in\set{F}_{1+\kappa}$ are PRUE flows that minimize $\Phi(\vec{\cdot})$ over their corresponding feasible sets, we have
\[
\Phi({\vec{v}}^0) \le \Phi \Big( \frac{\hat{\vec{v}}^0}{1+\kappa} \Big) \quad \text{and} \quad 
\Phi(\hat{\vec{v}}^0) \le \Phi \big( (1+\kappa)\vec{v}^0 \big),
\]
which imply
\begin{equation}\label{AA}
(1+\kappa)^{n+1} \sum_{a \in \Ac} \frac{b_a}{n+1} (v^0_a)^{n+1} \leq
\sum_{a \in \Ac} \frac{b_a}{n+1} (\hat{v}^0_a)^{n+1} 
\end{equation}
and
\begin{equation}\label{BB}
\sum_{a \in \Ac} \frac{b_a}{n+1} (\hat{v}^0_a)^{n+1} \leq
(1+\kappa)^{n+1} \sum_{a \in \Ac} \frac{b_a}{n+1} (v^0_a)^{n+1} 
\end{equation}
respectively.

Let us assume that $C(\vec{\hat{f}}^0) > (1+\kappa)^{n+1} C(\vec{f}^0)$,
which is equivalent to 
\begin{equation}\label{CC}
(1+\kappa)^{n+1} \sum_{a \in \Ac} b_{a}(v^0_a)^{n+1} <
\sum_{a \in \Ac} b_{a}(\hat{v}^0_a)^{n+1} .
\end{equation}
From $n\times \eqref{AA} + ((n+1)(1+\kappa)^n-1)\times\eqref{BB} + ((1+\kappa)^n-1)\times\eqref{CC}$, we obtain
\begin{equation} \label{DDD}
\theta_1 \sum_{a \in \Ac}\frac{b_{a}(v^0_a)^{n+1}}{n+1} 
< \theta_2 \sum_{a \in \Ac}\frac{b_{a}(\hat{v}^0_a)^{n+1}}{n+1}
\end{equation}
where
\begin{align*}
\theta_1 &=  n \cdot \frac{(1+\kappa)^{n+1} }{n+1} - ((n+1)(1+\kappa)^n-1) \cdot \frac{(1+\kappa)^{n+1}}{n+1} + ((1+\kappa)^n-1) \cdot (1+\kappa)^{n+1} = 0\\
\theta_2 &=  n \cdot \frac{1}{n+1} - ((n+1)(1+\kappa)^n-1) \cdot \frac{1}{n+1} + ((1+\kappa)^n-1) = 0
\end{align*}
for all $\kappa\geq0$.
Therefore, \eqref{DDD} leads to $0<0$, which is a contradiction.
We conclude that $C(\vec{\hat{f}}^0) \leq (1+\kappa)^{n+1} C(\vec{f}^0)$.
\end{proof}

\section{Nine-node Asymmetric Networks}\label{asynet}

In order to test the performance of $\UEPEX$ model in an asymmetric network, we create an asymmetric version of the nine-node network considered by \citet{hearn1998solving}.
In the asymmetric nine-node network, which has been shown in Figure \ref{fig6}, we add a few additional arcs and assume that the arc travel cost function is:
\begin{equation}\label{asyeq}
t_{a}(\vec{v})=A_{a}+B_{a}\bigg(\frac{0.5v_{\hat{a}} + v_{a}}{C_{a}}\bigg)^4
\end{equation}
where $\hat{a}$ is the flow in the opposite arc.
Thus, the arc travel function depends not only on the flow in that arc, but also on the flow in the arc in opposite direction.  
The values of parameters $A_{a}$, $B_{a}$ and $C_{a}$ are given in Table \ref{asynine} for each arc.
\begin{figure}\centering
	\begin{tikzpicture}[>=stealth',shorten >=1pt,auto,node distance=2.8cm]
		\node[shape=circle, draw] (1) at (0,0) {1};
		\node[shape=circle, draw] (5) at (4,0) {5};
		\node[shape=circle, draw] (7) at (8,0) {7};
		\node[shape=circle, draw] (3) at (12,0) {3};
		\path[->] (1) edge [bend left=10] node {} (5);
		\path[->] (5) edge [bend left=10] node {} (7);
		\path[->] (7) edge [bend left=10] node {} (3);
		\path[->] (5) edge [bend left=10] node {} (1);
		\path[->] (7) edge [bend left=10] node {} (5);
		\path[->] (3) edge [bend left=10] node {} (7);
		\node[shape=circle, draw] (2) at (0,-2) {2};
		\node[shape=circle, draw] (6) at (4,-2) {6};
		\node[shape=circle, draw] (8) at (8,-2) {8};
		\node[shape=circle, draw] (4) at (12,-2) {4};
		\node[shape=circle, draw] (9) at (6,-1) {9};
		\path[->] (2) edge [bend left=10] node {} (6);
		\path[->] (6) edge [bend left=10] node {} (8);
		\path[->] (6) edge [bend left=10] node {} (2);
		\path[->] (8) edge [bend left=10] node {} (6);
		\path[->] (4) edge [bend left=10] node {} (8);
		\path[->] (8) edge [bend left=10] node {} (4);
		\path[->] (1) edge [bend left=10] node {} (6);
		\path[->] (6) edge [bend left=10] node {} (1);
		\path[->] (2) edge [bend left=10] node {} (5);
		\path[->] (5) edge [bend left=10] node {} (2);
		\path[->] (5) edge [bend left=10] node {} (6);
		\path[->] (6) edge [bend left=10] node {} (5);
		\path[->] (7) edge [bend left=10] node {} (4);
		\path[->] (4) edge [bend left=10] node {} (7);
		\path[->] (7) edge [bend left=10] node {} (4);
		\path[->] (4) edge [bend left=10] node {} (7);
		\path[->] (7) edge [bend left=10] node {} (4);
		\path[->] (9) edge [bend left=10] node {} (7);
		\path[->] (6) edge [bend left=10] node {} (9);
		\path[->] (9) edge [bend left=10] node {} (6);
		\path[->] (5) edge [bend right=10] node {} (9);
		\path[->] (9) edge [bend right=10] node {} (5);
		\path[->] (9) edge [bend left=10] node {} (7);
		\path[->] (7) edge [bend left=10] node {} (9);
		\path[->] (8) edge [bend left=10] node {} (3);
		\path[->] (3) edge [bend left=10] node {} (8);
		\path[->] (6) edge [bend left=10] node {} (9);
		\path[->] (9) edge [bend left=10] node {} (6);
		\path[->] (8) edge [bend right=10] node {} (9);
		\path[->] (9) edge [bend right=10] node {} (8);
		\path[->] (7) edge [bend right=10] node {} (8);
		\path[->] (8) edge [bend right=10] node {} (7);
	\end{tikzpicture}
	\caption{Asymmetric nine-node network}
\label{fig6}
\end{figure}
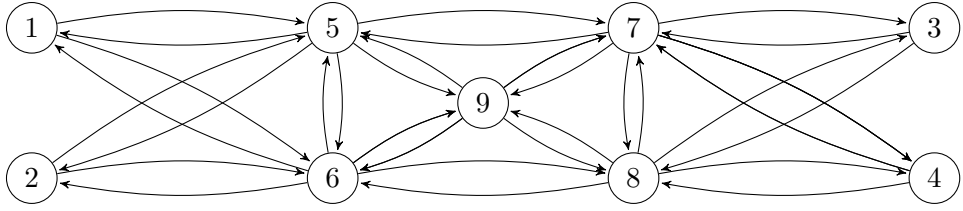

\begin{table}[]
\centering
\caption{asymmetric nine-node network arc cost function parameters}
\label{asynine}
\begin{tabular}{ c c c c }
\toprule
$a$  & $A_a$  & $B_a$    & $C_a$ \\
\midrule
(1,5) & 12 & 1.80 & 5 \\ 
(1,6) & 18 & 2.70 & 6 \\ 
(2,5) & 35 & 5.25 & 3 \\ 
(2,6) & 35 & 5.25 & 9 \\ 
(5,6) & 20 & 3.00 & 9 \\ 
(5,7) & 11 & 1.65 & 2 \\ 
(5,9) & 26 & 3.90 & 8 \\ 
(6,8) & 33 & 4.95 & 6 \\ 
(6,9) & 30 & 4.50 & 8 \\ 
(7,3) & 25 & 3.75 & 3 \\ 
(7,4) & 24 & 3.60 & 6 \\ 
(7,8) & 19 & 2.85 & 2 \\ 
(8,3) & 39 & 5.85 & 8 \\ 
(8,4) & 43 & 6.45 & 6 \\ 
(9,7) & 26 & 3.90 & 4 \\ 
(9,8) & 30 & 4.50 & 8 \\ 
(5,1) & 12 & 1.80 & 5 \\ 
(6,1) & 18 & 2.70 & 6 \\ 
(5,2) & 35 & 5.25 & 3 \\ 
(6,2) & 35 & 5.25 & 9 \\ 
(6,5) & 20 & 3.00 & 9 \\ 
(7,5) & 11 & 1.65 & 2 \\ 
(9,5) & 26 & 3.90 & 8 \\ 
(8,6) & 33 & 4.95 & 6 \\ 
(9,6) & 30 & 4.50 & 8 \\ 
(3,7) & 25 & 3.75 & 3 \\ 
(4,7) & 24 & 3.60 & 6 \\ 
(8,7) & 19 & 2.85 & 2 \\ 
(3,8) & 39 & 5.85 & 8 \\ 
(4,8) & 43 & 6.45 & 6 \\ 
(7,9) & 26 & 3.90 & 4 \\ 
(8,9) & 30 & 4.50 & 8 \\ 
\bottomrule
\end{tabular}
\end{table}

\end{document}